\documentclass[10pt]{article}
\usepackage[letterpaper, total={7in, 9in}]{geometry}
\usepackage[toc, page]{appendix}
\usepackage[utf8]{inputenc}
\usepackage{amsmath, mathtools}

\DeclareMathOperator*{\argmin}{arg\,min}
\usepackage{amssymb}
\usepackage{braket}
\usepackage{amsthm}
\usepackage{float}
\usepackage{graphicx}
\usepackage{enumerate}
\DeclareMathOperator{\Tr}{Tr}
\usepackage{mathtools}
\usepackage{comment}
\usepackage{xfrac}
\usepackage{nicefrac}

\usepackage{amsmath, amssymb, xspace,url}
\usepackage{epsfig}
\usepackage{longtable}
\usepackage[usenames, dvipsnames]{color}

\usepackage[pdftex,dvipsnames]{xcolor}
\usepackage{todonotes}
\usepackage[framemethod=tikz]{mdframed}
\usepackage{lineno}
\newmdenv[leftmargin=\dimexpr-0.4em, innerleftmargin=0.5em,
rightmargin=\dimexpr-0.4em, innerrightmargin=0.5em,
linewidth=2pt,linecolor=red, topline=false, bottomline=false,
innertopmargin=0pt,innerbottommargin=0pt,skipbelow=0pt,skipabove=0pt,%
]{notex}

\newenvironment{note}%
{\vskip\dimexpr\dp\strutbox-\prevdepth\relax\notex\strut\ignorespaces}%
{\xdef\notetpd{\the\prevdepth}\endnotex\vskip-\notetpd\relax}

\let\oldtodo\todo

\makeatletter%
\DeclareDocumentCommand{\todo}{ O{} +g +d<> }{%
	\IfNoValueTF{#2}{\relax}{%
		\oldtodo[caption={#2},size=\footnotesize,#1]{\renewcommand{\baselinestretch}{1}\selectfont\sffamily#2\par}%
	}%
	\IfNoValueTF{#3}{\relax}{%
		\IfNoValueTF{#2}{
			\begin{note}%
				\begin{internallinenumbers}%
					\indent%
					#3%
				\end{internallinenumbers}%
			\end{note}%
		}{
			\vspace{-0\baselineskip}%
			\begin{note}%
				\begin{internallinenumbers}%
					\indent%
					#3%
				\end{internallinenumbers}%
			\end{note}%
		}%
	}%
}%
\makeatother
\usepackage{soul}

\def\Cbb{{\mathbb{C}}}

\def\Pbb{{\mathbb{P}}}
\def\Qbb{{\mathbb{Q}}}

\def\Ebb{\mathbb{E}}

\def\Qbf{{\bf Q}}

\def\Ubf{{\bf U}}

\def\Mbf{{\bf M}}

\def\Mbf{{\bf M}}

\def\Ibf{{\bf I}}

\def\t{^\top}

\def\bkE{{\rm I\kern-.17em E}}
\def\bk1{{\rm 1\kern-.17em l}}
\def\bkD{{\rm I\kern-.17em D}}
\def\bkR{{\rm I\kern-.17em R}}
\def\bkP{{\rm I\kern-.17em P}}

\def\bkZ{{\bf{Z}}}
\def\bkE{{\rm I\kern-.17em E}}
\def\bk1{{\rm 1\kern-.17em l}}
\def\bkD{{\rm I\kern-.17em D}}
\def\bkR{{\rm I\kern-.17em R}}
\def\bkP{{\rm I\kern-.17em P}}
		\def\bkE{{\rm I\kern-.17em E}}
		\def\bk1{{\rm 1\kern-.17em l}}
		\def\bkD{{\rm I\kern-.17em D}}
		\def\bkR{{\rm I\kern-.17em R}}
		\def\bkP{{\rm I\kern-.17em P}}
		\def\bkY{{\bf \kern-.17em Y}}
		\def\bkZ{{\bf \kern-.17em Z}}
		\def\bkC{{\bf  \kern-.17em C}}

\def\Bscr{{\cal B}}

\def\Dscr{{\cal D}}

\def\Jscr{{\cal J}}
\def\Lscr{{\cal L}}
\def\Hscr{{\cal H}}

\def\Pscr{{\cal P}}
\def\Qscr{{\cal Q}}
\def\Sscr{{\cal S}}

\def\Uscr{{\cal U}}
\def\Vscr{{\cal V}}

\def\Nscr{{\cal N}}

\def\eef{\;\textrm{if}\;}

\let\forallnew\forall
\renewcommand{\forall}{\forallnew\ }
\let\forall\forallnew

\def\b12{(\beta_1,\beta_2)}

\newcommand{\Real}{\ensuremath{\mathbb{R}}}
\newcommand{\inv}{^{-1}}

\def\dim{\mathop{\hbox{\rm dim}}}  
\def\exp{\mathop{\hbox{\rm exp}}}

\def\half  {{\textstyle{1\over 2}}}

\def\conv{{\rm conv}}

\def\hbar{\skew{4.2}\bar h}

\def\superstar{^{\raise 0.5pt\hbox{$\nthinsp *$}}}
\def\SUPERSTAR{^{\raise 0.5pt\hbox{$*$}}}

\def\lamstarT {\lambda^{\raise 0.5pt\hbox{$\nthinsp *$}T}}

\newlength{\noteWidth}
\long\def\notes#1{\ifinner
{\tiny #1}
\else
\marginpar{\parbox[t]{\noteWidth}{\raggedright\tiny #1}}
\fi\typeout{#1}}
 \def\notes#1{\typeout{read notes: #1}} 

\newcommand{\ie}{i.e.\@\xspace} 
\newcommand{\eg}{e.g.\@\xspace} 

\def\spose#1{\hbox to 0pt{#1\hss}}

\def\text #1{\hbox{\quad#1\quad}}

\makeatletter
\newcommand{\pushright}[1]{\ifmeasuring@#1\else\omit\hfill$\displaystyle#1$\fi\ignorespaces}
\newcommand{\pushleft}[1]{\ifmeasuring@#1\else\omit$\displaystyle#1$\hfill\fi\ignorespaces}
\makeatother

\def\nthinsp{\mskip -2   mu}

\usepackage{ifthen}
\newboolean{showcomments}
\setboolean{showcomments}{true}
\newcommand{\aak}[1]{  \ifthenelse{\boolean{showcomments}}
{ \textcolor{blue}{(AAK says:  #1)}} {}  }

\newtheorem{theorem}{Theorem}[section]

\newtheorem{proposition}[theorem]{Proposition}

\newcounter{example}
\renewcommand{\theexample}{\thesection.\arabic{example}}

\newcounter{remark}
\renewcommand{\theremark}{\thesection.\arabic{remark}}

{\begin{list}{}%
         {\setlength{\leftmargin}{#1}}%
         \item[]%
}
{\end{list}}

		\def\bsp{\begin{split}}
		\def\beq{\begin{eqnarray}}
		\def\bal{\begin{align*}}
		\def\bc{\begin{center}}
		\def\be{\begin{enumerate}}
		\def\bi{\begin{itemize}}
		\def\bs{\begin{small}}
		\def\bS{\begin{slide}}
		\def\ec{\end{center}}
		\def\ee{\end{enumerate}}
		\def\ei{\end{itemize}}
		\def\es{\end{small}}
		\def\eS{\end{slide}}
		\def\eeq{\end{eqnarray}}
		\def\eal{\end{align*}}
		\def\esp{\end{split}}
		\def\qed{ \vrule height7.5pt width7.5pt depth0pt}  

	\def\cp2problem#1#2#3#4{\fbox
		 {\begin{tabular*}{0.9\textwidth}
			{@{}l@{\extracolsep{\fill}}l@{\extracolsep{6pt}}l@{\extracolsep{\fill}}c@{}}
				#1 & & $#4 $ 
			\end{tabular*}}}

		\renewcommand{\emph}[1]{\textbf{#1}}

		\def\bkE{{\rm I\kern-.17em E}}
		\def\bk1{{\rm 1\kern-.17em l}}
		\def\bkD{{\rm I\kern-.17em D}}
		\def\bkR{{\rm I\kern-.17em R}}
		\def\bkP{{\rm I\kern-.17em P}}
		
		\def\bkZ{{\bf{Z}}}

\newcommand {\beeq}[1]{\begin{equation}\label{#1}}
\newcommand {\eeeq}{\end{equation}}
\newcommand {\bea}{\begin{eqnarray}}
\newcommand {\eea}{\end{eqnarray}}

\def\texitem#1{\par\smallskip\noindent\hangindent 25pt
               \hbox to 25pt {\hss #1 ~}\ignorespaces}

\usepackage{authblk}
\reversemarginpar
\title{\bf The Quantum Advantage in Decentralized Control}
\author[1]{Shashank A. Deshpande\thanks{Corresponding author;\ MSC2020 : 93E20}} 
 \author[1]{Ankur A. Kulkarni} 
\affil[1]{Systems and Control Engineering,  Indian Institute of Technology Bombay, Mumbai 400076 \\ Emails: \texttt{shashank.deshpande.phy@iitb.ac.in, kulkarni.ankur@iitb.ac.in}}
\date{}

\begin{document}
\maketitle

\begin{abstract}
	It is known in the context of decentralised control that there exist control strategies consistent with the requirements of a given information structure, yet physically unimplementable through any amount of passive common randomness. This imposes a natural set of limitations on what is achievable through common randomness in both cooperative and competitive settings. We show that it is possible to breach these limitations with the use of quantum-physical architectures. In particular, we present a class of stochastic strategies that leverage quantum entanglement to produce strategic distributions which compose a strict superclass of strategies implemented through passive common randomness.  We investigate numerically, the `quantum advantage' offered by this new class over a parametric familiy of cooperative decision problems with static information structure. We demonstrate through variations across the parametric family that fundamental decision theoretic elements such as information and the cost determine the manifestation of quantum advantage in a given control problem. Our work motivates a novel decision and control paradigm with an enlarged space of control policies achievable by means of quantum architectures.
\end{abstract}
\section{Introduction}
In a decentralised control problem, each agent chooses an action conditioned on his information. A control strategy for the problem is thus a conditional distribution on the joint action space given the observations of all agents.  The information structure of the problem imposes further constraints on the control strategy and constrains it within a subset of the corresponding probability simplex. 

A widely studied space of strategies comprises of `mixed strategies' (common nomenclature in game theory literature) or strategies implemented using \textit{passive}, or pre-play common randomness. Here, all the agents have access to an externally provided random variable which does not correlate with any piece of information in the problem, in addition to any local randomization that the players may perform. Such strategies are physically implementable within the realm of classical physics, and  
hence, we also often refer to this strategic space as the space of \textit{classical strategies}.

Ananthram and Borkar \cite{ananthram2007commonrandom} show that for a problem with static information structure, strategies implemented using passive common randomness in general comprise a \textit{strict subset} of what is allowed by the information structure. In particular, in a static information structure, agents are only bound by the prohibition of in-play communication. Strategies that respect this are characterized by the conditional independence of the action of an agent and the observations of other agents, given the observation of the said agent. Prima-facie, any strategy satisfying this requirement ought to be allowed by a static information structure; we call such strategies \textit{no-signalling strategies}. 
Ananthram and Borkar \cite{ananthram2007commonrandom} show that set of no-signalling strategies is strictly larger than the strategies that the set of classical strategies. 
They perceive this strict inclusion to be a limitation of common randomness in decentralised control. Indeed, for a cooperative problem where the solution concept is the optimal cost, lower costs may be achievable for a given problem if access is allowed to the space beyond classical strategies, even while not violating the information structure. On the other hand, in competitive settings where the solution concept is the Nash equilibrium, appearance of new equilibria is possible (in team v/s team settings) as the agents access strategies beyond the set of mixed strategies. We refer the reader to \cite{ananthram2007commonrandom} for more.

The above limitation of passive common randomness is the starting point for our research. While passive common randomness can be implemented easily, there appears to be no general-purpose physically feasible way of implementing arbitrary no-signalling strategies. One is thus left with an unsatisfactory gap between what seems to be allowed by the information structure and what one can physically implement.
In this article, we demonstrate how quantum entanglement can be utilized as a resource to breach these limitations of common randomness. We show that entanglement assisted control architectures allow access to strategies beyond those achievable with the use of common randomness, without violating the information structure, which thereby allows cost improvements. In particular, we investigate a parametric family of static team decision problems where two decentralised agents collectively attempt to estimate an environmental variable based on their independent noisy observations of the variable. We find that while common randomness is ineffective at producing a better estimate through collaboration, entanglement enables a more complex collaboration among the agents, despite the lack of communication between them, and results in a strictly lower estimation error. We call this strict improvement, the \textit{quantum advantage}. In essence, the resulting \textit{spooky coordination at a distance} arising out of entanglement enlarges the space of joint distributions attainable by the agents beyond that attainable through classical strategies and thereby results in a lower cost.


Remarkably, our analysis  also reveals that entanglement assistance is not always useful for producing strict improvements over common randomness. We find that in our family of estimation problems, the agents' observations must be sufficiently informative of the natural state for a quantum resource to further enhance their classical estimate. Further, if one of the agents has an observation which is asymmetrically informative by more than a certain threshold, then the quantum advantage also disappears after this threshold.  Another decision-theoretic feature that characterizes the appearance of a quantum advantage is a parameter in the cost that captures the relative importance of correctly estimating different values of the environmental state. We find that the quantum advantage disappears if the parameter exits a certain range of values.

A cooperative decision problem presents itself as an optimisation of a linear objective in the conditional distributions (or occupation measures) subject to constraints imposed by the information structure (see, \eg,~\cite{kulkarni2014optimizer}). We find that for a static problem, the space accessible classically through common randomness as well as the set of no-signalling strategies are both polytopes, with the former strictly included within the latter. The space of quantum strategies is an open convex non-polytopic structure sandwiched between the two. 
We showcase some features of this relative geometry which provides for a geometric intuition to the fluctuations in the quantum advantage through the parametric family we examine.

\subsection{Randomization as a resource in decentralized control}
Though stochastic control problems are usually posed in the policy space, \eg, \cite{bertsekas00dynamic,witsenhausen_counterexample_1968,ho1980team,mahajan2012information} lifting to randomized strategies has often served a useful analytical tool~\cite{borkar88convex,manne60linear}. For problems with nonclassical information structures, it was shown that nonclassicality of the information structure implies the nonconvexity of the equivalent problem in the space of occupation measures \cite{kulkarni2014optimizer}, and convex relaxation-based approaches were introduced to solve and bound problems \cite{kulkarni2014optimizer,jose2015linear,wu11optimal}. Similar approaches have also been applied to problems in finite-blocklength communication~\cite{jose2017linear,jose2018improvedSW,jose2018shannon,matthews2012linear}. These instances notwithstanding, the main property that has been employed in each of the instances above is that randomization or lifting of the problem from the policy space to the space of occupation measures leads to \textit{equivalent problem}. In other words randomization merely results in an equivalent \textit{analytical reformulation} of the original policy-level problem.

In contrast, the problem posed over the space of quantum strategies is \textit{not} a reformulation of the original problem, but rather a generalization of it. While classical randomization, via passive common randomness, only serves as an analytical tool with no performance benefits in terms of cost, quantum randomization based on entanglement brings with it a distinct new \textit{resource} for decision making that allows for better performance than that achieved in the policy space. Geometrically, the space of occupation measures achievable through quantum strategies (also called the quantum ellitope) has extreme points that do not lie in the classical region. This geometric fact is borne out through the Bell and CHSH inequalities~\cite{clauser1969chsh}; that this is geometric fact is indeed also physical was shown through experiments that were awarded the Nobel prize last year. We note that quantum mechanical relaxations as convex relaxations were proposed  for bounding decentralized control problems in a recent survey~\cite{saldi2022geometry}.

This viewpoint of thinking of randomization as a resource in decentralized control leads to natural new questions, answers to some of which our work has provided glimpses to. How does one measure the ``amount'' of this resource? How does the performance change if the decision makers are provided ``more'' of this resource? What is the best way to use this new resource? Are there problems where the additional resource is redundant and classical strategies are just as good as quantum ones?  It appears that quantum strategies are not uniformly useful in control and a deeper understanding of the relation of problem structure and quantum advantage is needed. 
While the woods we are walking in are lovely, dark and deep, there are miles to go before we can address any of these questions satisfactorily. This article is a attempt towards initiating such an investigation.


\subsection{Organization}
This article is organised as follows. In Section \ref{sec:prelim}, we brush up on the prerequisites necessary for further development in this article. In Section \ref{sec:commonrandlim}, we formulate our decentralised control problem and demonstrate a limitation of common randomness within its context. In Section \ref{sec:quantumstrats}, we formulate the set of quantum strategies and demonstrate the thereby attained breach of the aforementioned limitation. We also present here, some general properties of the set of quantum strategies and later conclude this section with a numerical investigation of the role of information and the cost-structure in the manifestation of a quantum advantage. We offer our concluding remarks in Section \ref{sec:conclude}.

\section{Preliminaries}\label{sec:prelim}
We have two sets of preliminaries for this article. This section deals brushes up on the relevant prerequisites on team decision problems. An additional set of preliminaries on quantum mechanics required for this article is included in Appendix \ref{app:qmprelim}.
\subsection{Notation}
 Let $ \Cbb $ be the set of complex numbers. Recall that any $ x \in \Cbb $ may be written as $ x=a+ib $ where $ a,b \in \Real $. By $ x^* $ we denote the complex conjugate of $ x $, \ie, $ x^* = a-i b $ and by $ |x|  $ we denote its magnitude, $ |x|=a^2+b^2$. We denote the negation of $a\in\{0, 1\}$ by $\overline a$. For $a, b\in\{0, 1\}$, we denote the binary addition or `XOR' of $a$ and $b$ by $a\oplus b$. Similarly, binary `OR' and `AND' take the forms $a \vee b$ and $a.b$ respectively.
 $\Pscr(\Xi)$ denotes the set of probability distributions on $\Xi$, and $\Pscr(\Uscr|\Xi)$ denotes the set of conditional probability distributions on $\Uscr$ given an element in $\Xi$. $\delta(x, y):=\delta_{xy}$ denotes the Kronecker delta with $\delta_{xy}=1$ if $x=y$ and 0 otherwise. For a tuple $(a_i)_i$, $a_{-j}$ represents the tuple $(a_i)_{i\neq j}$.

\subsection{Static team decision problems}\label{sec:tdecisionprobs}

A static team theoretic decision problem comprises of the following~\cite{ho1980team}.
\begin{enumerate}[I]
\item \textit{Agents:}	A set $ \Nscr $ of players or agents or decision makers.
\item \textit{State of Nature:} A random vector $\xi_W\in\Xi_W$ with a known distribution $\Pbb(.)\in\Pscr(\Xi_W)$ denoting the \textit{state of nature}, also known as the \textit{environmental randomness}. $ \xi _W$ comprises of all the exogenous sources of noise in the system, including initial state, measurement noise and system noise.
\item\textit{Information:} Each decision maker $ i $ receives an imperfect measurement of the random vector $ \xi_W$. The information of decision maker $ i $ is denoted 
$$\xi_{i}:=\eta_i(\xi_W), \qquad\forall i\in\Nscr .$$
for some measurable function $ \eta_i $. We denote $ \Xi_{i} $ to be the space of $ \xi_{i}, $ whereby $ \eta_i: \Xi_W \rightarrow \Xi_{i}. $
Notice that $ \eta_i $ is a function only of $ \xi_W$ and is not affected by actions of the other decision makers. Thus we call this a problem of \textit{static information structure}. The joint conditional of $ \xi:=(\xi_{i})_{i\in\Nscr}$ given $ \xi_W$ is denoted $ \Pbb(\xi| \xi_W)$.
\item\textit{Action:} As a function of the information $ \xi_{i} $, the decision maker $ i $ is required to choose an action $u_i$ from a set $ \Uscr_i $. We assume each $ u_i $ to be a scalar without loss of generality, as a $d$ dimensional vector can be decomposed into $d$ separate decision makers.
\item\textit{Deterministic (pure) strategy:} The function  $\gamma_i: \Xi_{i}\to \Uscr_i$ that maps $ \xi_{i} $ to $ u_i $ is called a pure strategy, deterministic strategy, or simply \textit{strategy}, of player $ i $. \ie, we have,
$$u_i=\gamma_i(\xi_{i}),\qquad \forall i\in\Nscr$$
Let $ \Gamma_i $ be the set of functions $ \gamma_i. $
\item\textit{The cost function:} A function $\ell: \prod_{i\in\Nscr} \Uscr_i \times \Xi_W \to \mathbb{R}$ forms the cost function 
The goal is to find $\gamma :=(\gamma_i)_{i\in\Nscr} \in \Gamma :=\prod_{i\in\Nscr} \Gamma_i$ to solve,
\[ J^*_\Gamma := \min_{\gamma} J(\gamma), \]
where
$$J({\gamma}):=\Ebb [\ell(\gamma_1(\eta_1(\xi_W)),\hdots,\gamma_n(\eta_n(\xi_W)),\xi_W)].$$
\end{enumerate}
Notice that \textit{time} has no bearing on the problem owing to the static information structure. Also, for simplicity we will limit ourselves to the case where $ \Xi_W$, and all $\Uscr_i$ are finite.

We now discuss the spaces of strategies agents are allowed to play within the above information structure.
\subsubsection{Control strategies assisted by passive common randomness}
In the abovementioned space of deterministic strategies, $\Gamma$, each agent fixes a `strategy' $\gamma_i:\Xi_i\to\Uscr_i$ following which his action $u_i$ is produced as the function of his observation $\gamma_i(\xi_i)$. The tuple $\gamma=(\gamma_i)_{i\in\Nscr}$ can be thought of as a `local' deterministic strategy of the team. A natural extension is to allow the agents to locally randomize their choice of  action. In this space of `locally' randomized strategies (or as in the language of game theory, `behavioural' strategies), $\Dscr$, each agent $i\in\Nscr$ executes an action $u_i\in\Uscr_i$ sampled from  a conditional distribution $Q_i\in\Pscr(\Xi_i|\Uscr_i)$. The strategy $Q\in\Dscr$ is then expressible as 
\begin{equation}
	Q(u|\xi)=\prod_{i\in\Nscr} Q_i(u_i|\xi_i).
\end{equation}
A behaviourally randomized strategy produces conditionally `uncorrelated' actions given local observations as evident in the multiplicative separation above. A scheme to obtain correlations beyond this conditional independence is through externally provided passive common randomness. Herein, each agent $i\in \Nscr$ chooses his action $u_i$ randomly from a conditional $Q_i\in\Pscr(\Uscr_i|\Xi_i, \Omega)$, \ie  through a conditioning also on a random variable $\omega\sim \Phi(.)\in \Pscr(\Omega)$ in addition to his observation $\xi_i$. Here, $\Omega$ is assumed to be some discrete space of variables. A crucial assumption central to implementability of these strategies is the \textit{passivity} of the external randomness: $\omega$ is conditionally independent of all information in the problem, \ie, $\xi$ and $\xi_W$. We call this space of strategies the \textit{local polytope} $\Lscr$ where we will soon establish the polytopic nature of its geometry through Theorem \ref{thm:detoptimum}. First we emphasize that each joint strategy $Q\in\Lscr$ is expressible as
\begin{equation}\label{eq:commonrandom}
	Q(u|\xi)=\sum_{\omega\in\Omega} \Phi(\omega)\prod_i Q_i(u_i|\xi_i).
\end{equation}
Second, we notice that since  $\Phi(\omega)\geq 0$ for all $\omega$ and $\sum_{\omega\in\Omega} \Phi(\omega) =1$, we have the relation $\Lscr=\conv(\Dscr)$ by definition where $\conv(\bullet)$ denotes the convex hull of `$\bullet$' and we recall that $\Dscr$ is the space of locally random strategies. Under a joint strategy $Q$, the expected cost of a decision problem $D$ is then given by
\begin{equation}\label{eq:expcostD}
J(Q;D)=\sum_{u, \xi, \xi_W}\Pbb(\xi)\ell(u, \xi_W)Q(u|\xi).
\end{equation}
Hereon given a strategic space $\Sscr$ for a problem $D$, we denote  $J_{\Sscr}^*(D)=\inf_{Q\in\Sscr} J(Q;D)$.

\subsubsection{No-signalling strategies}
Any strategy $Q\in \Pscr(\Uscr|\Xi)$ obeys the positivity and normalisation constraints $Q(u|\xi)\geq 0\ \forall\ u, \xi$,   $\sum_u Q(u|\xi)=1\ \forall\ \xi$. In addition, the decentralised, static information structure of the problem imposes the `no-signalling' constraints on $Q$. In particular it is required that any agent's action is conditionally independent of the observation by another agent, given his own observation, i.e. $Q(u_i|\xi_i, \xi_{-i})=Q(u_i|\xi_i)$ for each $i\in \Nscr$. This conditional independence can be captured through the following linear equalities for each $i\in\Nscr$:
\begin{equation}\label{eq:nosignallingconst}
	\sum_{u_{-i}} Q(u_i, u_{-i}|\xi_i, \xi_{-i})=\sum_{u_{-i}} Q(u_i, u_{-i}|\xi_i, \xi'_{-i})\ \qquad \forall u_i, \xi_i, \xi_{-i}, \xi'_{-i}.
\end{equation} 
This asserts that the choice of conditional distribution of one agent given his information does not affect the outcome distribution of other agents, and thus the joint distribution respects the prohibition of \textit{super-luminal communication}. Since this set of distributions is characterised by a finite number of linear equalities and inequalities, it is indeed a polytope. We denote this set by $\mathcal{NS}$, and call it the `no-signalling' polytope. Strategies within the no-signalling polytope respect the static information structure in the problem since the no-signalling conditions \eqref{eq:nosignallingconst} prohibit any communication.

\subsubsection{Limitations of common randomness}\label{compdet}
A static decision problem considerably simplifies within the introduced strategic classes $\Gamma$, $\Lscr$ and $\Dscr$. In particular, the optimal cost over the whole union of these classes is attained an optimal deterministic strategy. This follows from the theorem presented below.
\begin{theorem}\label{thm:detoptimum}
i) The space of behaviourally random strategies is contained within the local polytope \ie $\Dscr\subseteq \Lscr$.\\
ii) Thereby, $J_{\Lscr}^*(D)=J_{\Dscr}^*(D)=J_{\Gamma}^*(D)$.
\end{theorem}
\begin{proof}
\textit{i)}	In \eqref{eq:commonrandom}, taking $|\Omega|=1$ recovers a strategy in $\Dscr$, thereby establishing that $\Dscr \subseteq \Lscr.$\\
\textit{ii)} Notice from \eqref{eq:expcostD} that the expected cost $J(Q;D)$  is a linear objective in $Q$. 
Since $\Lscr=\conv(\Dscr)$, we have $J^*_\Lscr(D)=J^*_\Dscr(D).$ Now the optimization over $\Dscr$ is a multilinear optimization, whose minimum is attained at an extreme point, \ie, in $\Gamma$ (see~\cite{kulkarni2014optimizer} for a detailed argument). Thus $J^*_\Gamma(D)=J^*_\Dscr(D)$ and the claim follows.
\end{proof}
Hence, even within the availability of common randomness, it is optimal for a team to play deterministically based on players' local information. 

There may, however, be a gap between $J^*_{\Nscr\Sscr}(D)$ and $J^*_\Lscr(D)$, pointing to a limitation of common randomness. We highlight this limitation of common randomness for static teams in the following section through an explicit example of a team problem. In particular, we show that there are no-signalling strategies that ourperform the deterministic optimality. This is precisely the limitation presented by Ananthram and Borkar \cite{ananthram2007commonrandom}.

\section{Problem Formulation}\label{sec:commonrandlim}
We now focus on a specific instance of decentralised estimation where two agents $B$ and $H$ attempt to jointly estimate the environmental state $\xi_W\sim Unif[\Xi_W]$ given independent, identically distributed noisy observations of the same. Let $\xi_B\in \Xi_B, \xi_H\in \Xi_H$ be their respective observations and $u_B\in \Uscr_B, u_H\in\Uscr_H$ be their actions.  
Each agent $i\in\{B, H\}$ picks his action $u_i$ given his observation $\xi_i$.  For simplicity of demonstration, we work with $\Xi_W=\Xi_B=\Xi_H=\Uscr_B=\Uscr_H=\{0,1\}$. The cost function of the problem is that of a decentralised estimation problem with a reward for producing a correct joint estimate of the environmental state, and is asymmetric in the value of this state. Let $f(u_B, u_H)$ be the joint estimate of the agents, where $f:\Uscr_B\times\Uscr_H\to \{0,1\}$. We take the following cost function,
\begin{equation} \label{eq:ldef}
\ell(u_B, u_H, \xi_W)=-\chi(\xi_W)\delta(\xi_W,f(u_B, u_H)),
\end{equation}
where $\chi(\xi_W) \in [0,\infty)$ weighs the estimation error as a function of the environmental state $\xi_W.$ 
We work with a symmetric, non-degenerate aggregator $f$ which obeys $f(a,b)=f(b,a)$ and $f(a,b)\neq f(\sim a, b)$. Upto relabelling of actions, it suffices to work with $f(u_B, u_H)=u_B\oplus u_H$. The joint distribution of $\xi_B,\xi_H,\xi_W$ is specified as follows through parameters $1/2\leq \lambda_i\leq 1$ for each $i\in\Nscr$, 
\begin{equation}\label{eq:priordist}
	\Pbb(\xi_B, \xi_H, \xi_W)=\Pbb(\xi_W)\Pbb(\xi_B|\xi_W)\Pbb(\xi_H|\xi_W);\qquad  \Pbb(\xi_i|\xi_W)=\begin{cases}
		\lambda_i & \xi_i=\xi_W\\
		1-\lambda_i & \xi_i\neq\xi_W
	\end{cases}\ \forall\ i\in\{B, H\}.
\end{equation}

\noindent The expected cost of our problem `$D\equiv (\lambda_B, \lambda_H, \chi(0), \chi(1))$' under a strategy $Q\in \Pscr(\Uscr_B\times\Uscr_H|\Xi_B\times\Xi_H)$ is thus given by
\begin{equation}
	J(Q;D)=-\sum_{u, \xi}\chi(\xi_W)\Pbb(\xi_B, \xi_H, \xi_W)\delta(\xi_W, u_B\oplus u_H) Q(u_B, u_H|\xi_B, \xi_H).
\end{equation}

\subsection{Deterministic optimal strategy}
From Theorem \ref{thm:detoptimum}, lowest cost achievable by common randomness, \ie, over $\Lscr$ is simply the optimal cost over deterministic strategies $\Gamma$. 
The following proposition characterizes the optimal deterministic strategy. 
\begin{proposition}\label{prop:classicalsol}
	Let $\lambda_H \geq \lambda_B$ and let
	\begin{equation}\label{eq:optclascost}
		\widehat{J}:=\min(-\chi(0)\Pbb(0), -\lambda_H\sum_{\xi_W}\Pbb(\xi_W)\chi(\xi_W), -\chi(1)\Pbb(1)).
	\end{equation}
	Then an optimal deterministic strategy for the decision problem above is given by:
	\begin{align}\label{eq:optclasstrat}
		\gamma^*_B(\xi_B)\equiv 0; \	\gamma^*_H(\xi_H)=\begin{cases}
			0&   \eef \widehat{J}=-\Pbb(0)\chi(0,)\\
			1&  \eef \widehat{J}=-\Pbb(1)\chi(1),\\
			\xi_H& \eef \widehat{J}=-\lambda_H\sum_{\xi_W} \Pbb(\xi_W)\chi(\xi_W).\\
		\end{cases}
	\end{align}
	Moreover $J^*_\Lscr(D) = \widehat{J}.$ 
\end{proposition}
\begin{proof}
Consider an arbitrary deterministic strategy $\gamma \in \Gamma$, and let $u_i:=\gamma_i(\xi_i)$ for an agent $i\in\{B, H\}$ and $\xi_i \in \Xi_i $. It is easy to check by direct evaluation that if $\gamma\equiv \gamma^*$, then $J(\gamma; D)=\widehat J$. Now suppose that $\gamma_j(0)\neq \gamma_j(1)$ for some $j\in\{B,H\}$. Denote $\delta(u_{-j}, \gamma_{-j}(\xi_{-j}))\Pbb(\xi_{-j}|\xi_W)=:R(u_{-j}, \xi_{-j}|\xi_W)$. Clearly, $\sum_{u_{-j,} \xi_{-j}} R(u_{-j}, \xi_{-j}|\xi_W)=1$. Now notice that
\begin{align*}
	J(\gamma; D)
	&=\sum_{\xi_W}\Pbb(\xi_W) \sum_{\xi_j,\xi_{-j}, u} \ell(u, \xi_W) \Pbb(\xi_j|\xi_W) \delta(u_j, \gamma_j(\xi_j))\times R(u_{-j}, \xi_{-j}|\xi_W).
\end{align*}
Now summing over $\xi_j$ and using $\Pbb(\xi_j=\xi_W|\xi_W)=\lambda_j$ gives
\begin{align*}
	J(\gamma;D)&= \sum_{\xi_W, u, \xi_{-j}}\Pbb(\xi_W)\left [\lambda_j\delta(u_j, \gamma_j(\xi_W))\ell(u, \xi_W) +(1-\lambda_j)\delta(u_j, \gamma_j(\overline{\xi_W}))\ell(u, \xi_W)\right]R(u_{-j}, \xi_{-j}|\xi_W)\\
	&= \sum_{\xi_W, u_{-j}, \xi_{-j}}\Pbb(\xi_W)\left [\lambda_j\ell(\gamma_j(\xi_W), u_{-j}, \xi_W)  +(1-\lambda_j)\ell(\gamma_j(\overline{\xi_W}), u_{-j}, \xi_W)\right]R(u_{-j}, \xi_{-j}|\xi_W),
\end{align*}
where we have used the definition of $\delta(\cdot,\cdot)$ to sum over $u_j$. Now from \eqref{eq:ldef}, and using that $\gamma_j(\xi_W) \neq \gamma_j(\overline{\xi_W})$ observe that the term inside the square brackets
must take one of the two values $\{-\lambda_j\chi(\xi_W), -(1-\lambda_j)\chi(\xi_W) \}$ for each $u_{-j},\xi_W.$ Since $\lambda_j \geq \half$, it follows that $-\lambda_j\chi(\xi_W)\leq -(1-\lambda_j)\chi(\xi_W)$ and hence,
\begin{align*}
	J(\gamma;D)\geq -\sum_{\xi_W, u, \xi_{-j}}\Pbb(\xi_W)R(u_{-j}, \xi_{-j}, \xi_W)\lambda_j \chi(\xi_W) = -\lambda_j\sum_{\xi_W}\chi(\xi_W)\Pbb(\xi_W) \geq -\lambda_H\sum_{\xi_W}\chi(\xi_W)\Pbb(\xi_W),
\end{align*}
where we have used that $\lambda_B \leq \lambda_H.$
This shows that for each deterministic $\gamma$ for which $\exists j\in\{B, H\}$ satisfying $\gamma_j(0)\neq\gamma_j(1)$, $J(\gamma;D)\geq \widehat J$. Otherwise if $\gamma$ is such that $\gamma_i(0)\equiv\gamma_i(1)$ for all $i\in\{0, 1\}$, then $u_B\oplus u_H\equiv c\in\{0, 1\}$ and consequently, $J(\gamma; D)\in \{-\chi(0)\Pbb(0), -\chi(1)\Pbb(1)\}$ so that $J(\gamma; D)\geq \widehat J$. Since the cost $\widehat{J}$ is attained by $\gamma^*$, we have 
$\widehat J=\min_{\gamma\in\Gamma} J(\gamma; D)$, and it follows that $\widehat J=J_\Lscr(D)$ from Theorem \ref{thm:detoptimum}. 
\end{proof}
\def\NS{{\Nscr\Sscr}}
We make the following observation about the optimal strategy above. Notice that the optimal strategy corresponds to at most one agent actively estimating $\xi_W$. \ie, $J_{\Lscr}^*(D)$ is attained either when both agents play a constant strategy by disregarding the observed information (in which case $J^*_\Lscr(D)\in\{-\chi(0)\Pbb(0), -\chi(1)\Pbb(1)\}$) or only the more informed agent (agent $H$ in our case) putting his information to use and the other ignoring his information (in which case $J^*_\Lscr(D)= -\lambda_H\sum_{\xi_W}\Pbb(\xi_W)\chi(\xi_W)$). Thus, for this problem, there is no benefit in the more informed agent correlating his estimate with that of the other (less or equally informed) agent. This is applicable when optimizing over the space $\Gamma$ or indeed over $\Lscr$, where agents are provided access to passive common randomness. 

However, there do exist \textit{no-signalling} strategies in which the estimates of both agents are nontrivially correlated, and improve upon the cost $\widehat{J}$ attained above. We show this in the following section.

\subsection{Optimality as allowed by the information structure}

We now study the $\Nscr\Sscr$ polytope for this problem with the aim of characterizing $J^*_{\Nscr\Sscr}$. The geometry of this polytope is well-known for the problem setting we consider, \ie, $|\Uscr_i|=|\Xi_i|\equiv 2$. In particular, it is known \cite{barrett2005nonlocal} that the polytope has $24$ vertices, $16$ of which correspond to the determinstic strategies and $8$ correspond to what we refer to as the non-local vertices. The following proposition specifies the geometry of this polytope.
\begin{proposition}\cite[pg 3, eq (6)-(7)]{barrett2005nonlocal}
	The no-signalling polytope $\Nscr\Sscr$ has 24 vertices out of which:\\
a) The sixteen local vertices of the no-signalling polytope that correspond to the deterministic strategies are given by
\begin{equation}
	\pi^{\alpha\eta\beta\delta}(u_B, u_H|\xi_B, \xi_H)=\begin{cases}
		1 & u_B=\alpha.\xi_B\oplus \beta\\
		& u_H=\eta.\xi_H\oplus \delta\\
		0 & \text{otherwise}
	\end{cases} \text{ where } \alpha, \beta, \eta, \delta\in \{0,1\}.  
\end{equation} 
b) The eight non-local vertices are denumerable as
\begin{equation}\label{eq:nsvertex}
	Q^{\alpha\beta\delta}(u_B, u_H|\xi_B, \xi_H)=\begin{cases}
		1/2 & u_B\oplus u_H=\xi_B.\xi_H\oplus \alpha.\xi_B\oplus \beta.\xi_H\oplus \delta\\
		0 & \text{otherwise}
	\end{cases} \text{ where } \alpha, \beta, \delta\in\{0,1\}.
\end{equation}
\end{proposition} 
The proof of the above characterization involves a search algorithm over equality constraints \eqref{eq:nosignallingconst}. We refer to \cite{barrett2005nonlocal} for further details. 

The following proposition shows that there are instances of the problem $D$ where $J^*_\NS(D)< J^*_\Lscr(D).$ In particular, this implies that $\Lscr \subset \NS$ (a strict inclusion), a fact that was also seen in the counterexample presented by Ananthram and Borkar~\cite{ananthram2007commonrandom}.
\begin{proposition}\label{prop:nosigadvexists}
There exist $\chi(0), \chi(1)\in[0,\infty)$ and $\lambda_B, \lambda_H \in [0, 1]$ such that
\begin{equation}
	\min_{\alpha, \beta, \delta  \in\{0, 1\}} J(Q^{\alpha\beta\delta}; D)< \min_{\alpha,\beta,\gamma,\delta\in\{0, 1\}} J(\pi^{\alpha\beta\gamma\delta};D).
\end{equation}
Hence, $\exists\ \chi(0), \chi(1)\in\Real^+, \lambda_B, \lambda_H \in [0, 1]$ such that {$J^*_{\NS}(D)<J^*_\Lscr(D)$}.
\end{proposition}
\begin{proof}
We establish the proposition by provision of a direct numerical example. Take $\chi(0)=1, \chi(1)=3$ and $\lambda_B=\lambda_H=:\lambda=4/5$. Then, 
\begin{equation}
	J_\Lscr^*(D)=\min(-\chi(0)\Pbb(0),-\chi(1)\Pbb(1), -\lambda\sum_{\xi_W}\chi(\xi_W)\Pbb(\xi_W) )=-8/5.
\end{equation}
On the other hand, notice that for the non-local vertex $Q^{111}$ defined in \eqref{eq:nsvertex}, we have 
\begin{align}
 J(Q^{111};D)&=-{\chi(0)}{}\sum_{\xi_B, \xi_H} \Pbb(\xi_B, \xi_H, 0) (\xi_B.\xi_H\oplus \xi_B\oplus\xi_H)-{\chi(1)}{}\sum_{\xi_B, \xi_H} \Pbb(\xi_B, \xi_H, 1) (\xi_B.\xi_H\oplus \xi_B\oplus\xi_H)\\
& =-44/25<-8/5.
\end{align}
This establishes the proposed.
\end{proof}

Proposition \ref{prop:nosigadvexists} identifies a gap between the optimal cost of problem $D=(\nicefrac{4}{5}, \nicefrac{4}{5}, 1, 3)$ achievable by strategies implemented through common randomness, and the optimal cost as allowed by the static information structure of the problem. This gap arises because in the former space, optimality is characterised  by a strategy where a single agent, say $i$ produces an in-play prediction for $\xi_W$ conditioned on his local information while the other agent ignores his observations entirely. Subsequently this absence of in-play collaboration among the agents renders $\xi_{-i}$, the independently recorded observation of the other agent, futile. On the contrary, the strategy $Q^{111}$ is fundamentally stochastic, and  correlates the actions of the agents. This correlation contributes to their joint estimate being `right' more often. As an illustration, note that
\[ Q^{111}(u_B=0, u_H=0|\xi_B=1, \xi_H=1) = Q^{111}(u_B=1, u_H=1|\xi_B=1, \xi_H=1)=0 \]
\[Q^{111}(u_B=0, u_H=1|\xi_B=1, \xi_H=1)=Q^{111}(u_B=1, u_H=0|\xi_B=1, \xi_H=1) =\half. \]
On the other hand,
\begin{multline*}
Q^{111}(u_B=0|\xi_B=1, \xi_H=1) = Q^{111}(u_B=1|\xi_B=1, \xi_H=1)\\ =Q^{111}(u_H=0|\xi_B=1, \xi_H=1) = Q^{111}(u_H=1|\xi_B=1, \xi_H=1) = \half
\end{multline*}
whereby 
\[ Q^{111}(u_B=0, u_H=1|\xi_B=1, \xi_H=1) > Q^{111}(u_B=0|\xi_B=1, \xi_H=1)Q^{111}(u_H=1|\xi_B=1, \xi_H=1)\]
\[ Q^{111}(u_B=1, u_H=0|\xi_B=1, \xi_H=1) > Q^{111}(u_B=1|\xi_B=1, \xi_H=1) Q^{111}(u_H=0|\xi_B=1, \xi_H=1). \]
This stronger correlation\footnote{In fact, the correlation as been skewed away from the product of the marginals. One can see that $Q^{111}(00|00) < Q^{111}(0|00)Q^{111}(0|00)$, and likewise $Q^{111}(11|00) < Q^{111}(1|00)Q^{111}(1|00)$.} afforded by $Q^{111}$, even while remaining in the $\NS$ polytope, contributes to a lower estimation error.

Although a strategy in $\NS$ improves the cost of this instance, a physically method for realising arbitrary strategies in $\NS$ is unknown. As the central contribution of this article, we show that quantum mechanics allows for a physically realisable mechanism through which the agents can correlate their actions and improve upon the cost achieved through common randomness. This bridges some but not all of the gap between what is allowed by the static information structure $(J_{\Nscr\Sscr}^*(D))$ and what is achievable via passively random strategies $(J_{\Lscr}^*(D))$.

\section{Quantum Strategies for Decentralised Control}\label{sec:quantumstrats}
In the section that follows, we formulate  a class of stochastic strategies that exploit quantum entanglement to generate randomness. In contrast to the common randomness that has been the subject of \cite{ananthram2007commonrandom} and of our discussion so far, quantum strategies are not passively random, \ie they cannot be expressed in the form \eqref{eq:commonrandom}. Quantum randomness allows for richer correlations to manifest among the actions of the agents without violating the static information structure, and hence the set of quantum strategies $\Qscr$ extends beyond the local polytope $\Lscr$ but remains a subset of $\NS$. We also establish the convexity of the set $\Qscr$. 
\subsection{Formulation of $\Qscr$ -- the set of quantum strategies}\label{sec:quantum_strategy_formulation}
We now formulate a quantum strategy for a static team decision problem as described in Section~\ref{sec:tdecisionprobs}. Recall that $\Nscr$ is the set of agents in a given static decision problem and for each agent $i\in\Nscr$, we denote his observation by $\xi_i\in\Xi_i$ and action $u_i\in\Uscr_i$.  For readers not familiar with the quantum mechanical postulates, or associated notation, we recommend a quick glance onto Appendix \ref{app:qmprelim}.
\begin{itemize}
	\item The quantum strategy requires that the agents share a composite system in a state that is specified by a density matrix $\rho\in\Bscr(\otimes_{i}\Hscr_i)$. Each agent $i\in\Nscr$ has access to a physical subsystem described within a Hilbert space $\Hscr_i$ of dimension $\dim\Hscr_i=:d_i$.
   \item The agents agree upon their collective set of measurement and action strategies. For agent $i$, a measurement strategy $\beta_i: \Xi_i\to 2^{\Bscr(\Hscr_i)}$ specifies a choice of a \textit{POVM}  $\beta_i(\xi_i)=\{P_{a_i}^{i}(\xi_i)\}_{a_i\in A_i}$, where $A_i\subset \Real$ is finite, given the  observation $\xi_i\in\Xi_i$. A measurement strategy effectively specifies the measurement of an observable $\Mbf_i:=\sum_{a_i} a_i P_{a_i}^i
   (\xi_i)\in \Bscr(\Hscr_i)$, which results in an outcome $a_i\in A_i$.
   
  \item  An action strategy $\gamma_i:A_i\to\Uscr_i$ specifies the agent's action, given his measurement outcome. Given the freedom in the choice of the measurement stategy, it is without loss of generality to restrict $A_i= \Uscr_i$ and fix the action strategy to identity: $\gamma_i(a)=a$. To emphasize this simplification, we will henceforth shift our notation for the measurement observable from $\Mbf_i$ to $\Ubf_i$.  
  This simplification then informs that the set $\{P_{u_i}^i(\xi_i)\}_{u_i\in\Uscr_i}$ forms a \textit{POVM}, and $P_{u_i}^{i}(\xi_i)$ are projection operators (allowed to be null) which obey
   \begin{align}
   	 P_{u_i}^i(\xi_i)P_{u_i}^i(\xi_i)&=P_{u_i}^i(\xi_i)\ \forall\ i\in\Nscr, \xi_i\in\Xi_i, u_i\in\Uscr_i\\
   	 \sum_{u_i} P_{u_i}^i(\xi_i)&=\Ibf_i \ \forall \ i\in\Nscr, \xi_i\in\Xi_i
   \end{align}
where $\Ibf_i$ is the indentity operator on $\Hscr_i$.
\end{itemize}
Given a mathematical specification of the strategy $Q=(\{\Hscr_i\}_{i\in\Nscr}, \rho, \{P_{u_i}^i(\xi_i)\}_{i, \xi_i, u_i})$, the gameplay through which the strategy is realised is as follows
\begin{itemize}
\item Before gameplay agents share a composite quamtum system $\rho$ and agree on the measurement strategies $\beta_i, i\in \Nscr$.
\item In gameplay, the observations $\xi_i, i \in \Nscr$ are realised and observed by the agents
\item Observing $\xi_i\in\Xi_i$ each agent $i$ performs his local measurement on $\rho$ with the \textit{POVM} $\beta(\xi_i)=\{P_{u_i}^i(\xi_i)\}_{u_i\in\Uscr_i}\subset \Bscr(\Hscr_i)$. The \textit{POVM} corresponds to a measurement of the $\Uscr_i$-valued \textit{action} \textit{observable} $$\Ubf_i:=\sum_{u_i\in\Uscr_i} u_i P_{u_i}^i(\xi_i)\in\Bscr(\Hscr_i).$$
The precise of order measurements is fixed before gameplay, but since it does not affect the outcome, we do not specify it as part of the protocol.
\item  Each agent $i$ chooses the action $u_i$ resulting from his measurement
\end{itemize}

It follows from the measurement postulate for composite systems that for such a realization of the quantum strategy (See Appendix \ref{app:qmprelim}), the joint probability of the collective outcome $u=(u_i)_i\in\prod_i \Uscr_i$ given an observation tuple $\xi=(\xi_i)_i\in \prod_i\Xi_i$ is given by
\begin{equation}\label{eq:quantumoccumeasure}
	Q(u|\xi)=\Tr\left(\rho\bigotimes_{i\in\Nscr} P_{u_i}^i(\xi_i)\right)
\end{equation}
Notice that once all the agents have settled on a measurement strategy, the joint probability of the outcomes is unaffected by the order in which the measurements are physically performed. 
Hence, while the loss in simultaneity of measurements comes inevitably in a real world scenario, it does not affect the realization of a quantum strategy once the subsystems that constitute the composite state $\rho$ are assumed to be persistent to environmental corruption. This is because the order of measurements has no bearing on the joint probability of the outcomes -- a fact we demonstrate explicitly for two agents in Appendix \ref{app:qmprelim}.2, Postulate 5.

 Nevertheless, this discussion of the order of measurements is crucial. Indeed, since the measurement by each agent collapses the state of the composite system, the order of measurements affects the state upon which each local measurement is performed by the other agents. However, we do not consider this order among the strategic elements for it is inconsequential conditional probability defined in \eqref{eq:quantumoccumeasure}. We conclude this section by a re-emphasis on the mathematical formulation of  a quantum strategy, \ie, a quantum strategy $Q$ is (by overloading of notation) denoted by the tuple $(\{\Hscr_i\}_{i\in\Nscr}, \rho, \{P^i_{u_i}(\xi_i)\}_{i, \xi_i, u_i})$. We let the set of quantum strategies $\Qscr$ denote the set of $Q \in \Pscr(\Uscr|\Xi)$ that satisfy \eqref{eq:quantumoccumeasure} for some choice of Hilbert spaces $\Hscr_i$, a composite state $\rho$, and POVMs $\{P_{u_i}^i(\xi_i)\}_{u_i\in \Uscr_i}, \xi_i \in \Xi_{i},$ for each $i\in \Nscr.$ 

\subsection{Numerical demonstration of the quantum advantage}\label{sec:numdemo}
We now numerically demonstrate the cost advantage offered by the set of quantum strategies. 
We attend the problem $D=(\lambda_B, \lambda_H, \chi(0), \chi(1))$ that has been the subject of our investigation in Section \ref{sec:commonrandlim}. Recall that for $\chi(0)=1, \chi(1)=3$ and $\lambda_B=\lambda_H=\nicefrac{4}{5}$, $J^*_\Lscr(D)=\nicefrac{-8}{5}$ and $J_{\Nscr\Sscr}^*(D)=\nicefrac{-44}{25}$. We are primarily interested in demonstrating that $J^*_\Qscr(D)<J^*_\Lscr(D)$. To this end, we directly specify a quantum strategy $Q\in\Qscr$ for $D$ and show that it obtains a cost strictly lower than $J^*_\Lscr(D)$.

We now define a quantum strategy $Q=(\Hscr_B, \Hscr_H, \rho, \{P^B_{u_B}(\xi_B)\}_{u_B, \xi_B}, \{P^H_{u_H}(\xi_H)\}_{u_H, \xi_H})$ for the instance $D$ as follows.
\begin{enumerate}[i)]
	\item $\dim\Hscr_B=\dim\Hscr_H=2.$ We denote the identity on $\Hscr_B$ and $\Hscr_H$ by $\Ibf$.
	\item The density matrix $\rho\in\Bscr(\Hscr_B\otimes\Hscr_H)$ is given by
	\begin{equation}\label{eq:numstatematrix}
		\rho=\begin{pmatrix}
			{1}/{\sqrt{2}} \\0\\0 \\{1}/{\sqrt{2}}
		\end{pmatrix}\begin{pmatrix}
			{1}/{\sqrt{2}} &0&0&{1}/{\sqrt{2}}
		\end{pmatrix} =\begin{pmatrix}
	{1}/{{2}} &0&0&{1}/{{2}}\\
	0 &0&0&0\\
	0&0&0&0\\
	{1}/{{2}} &0&0&{1}/{{2}}
\end{pmatrix}
	\end{equation}
where we use some orthonormal basis of $\Hscr_B\otimes\Hscr_H$ for our representation. It is evident that $\Tr\rho=1$ and $\rho\succeq 0$ since $\rho=v v^T$ for a $v\in \Hscr_B\otimes\Hscr_H$. 
   \item To specify the measurement strategy, we first consider the following parameterised operator $P(\mu, a, b, \theta)$  in $\Hscr_B$ (or $\Hscr_H$) 
   \begin{equation}
    P(\mu, a, b, \theta)=\frac{1}{\mu}\begin{pmatrix}
    	a & e^{-\iota \theta}\\
    	e^{\iota \theta} & b
    \end{pmatrix} 
   \end{equation}
where $\mu,a,b \in \Real$ such that 
\begin{equation}\label{eq:squareequal}
{\mu=a+b, \mu a = 1+a^2, \mu b=1+b^2 , \text{ and }\theta \in [0,2\pi).}
\end{equation}
We also have a special notation $P(\infty, \infty, 0, 0):=\begin{pmatrix}
	1&0\\ 0&0
\end{pmatrix}=\lim_{\mu\to\infty} P(\mu, \mu, 0,0)$ and similarly $P(\infty, 0, \infty,\pi):=\begin{pmatrix}
0&0\\ 0&1
\end{pmatrix}= \lim_{\mu\to\infty} P(\mu, 0, \mu,\pi)$.
Now we specify the measurement strategy so as complete the specification of our $Q$:
\begin{align}
	P_0^B(0)&=P(\infty, \infty,0,0);\ &P_0^B(1)=P(\frac{4}{\sqrt{3}},\sqrt{3},\frac{1}{\sqrt{3}},\pi/2);\nonumber\\
		P_1^B(0)&=\Ibf-P_0^B(0)=P(\infty, 0, \infty, \pi);\ &P_1^B(1)=\Ibf-P_0^B(1)=P(\frac{4}{\sqrt{3}}, \frac{1}{\sqrt{3}}, \sqrt{3}, 3\pi/2);\nonumber\\ 
	 P_0^H(0)&=P(\frac{4}{\sqrt{3}},{\sqrt{3}},\frac{1}{\sqrt{3}},3\pi/2);\  &P_0^H(1)=P(\frac{4}{\sqrt{3}},\frac{1}{\sqrt{3}},\sqrt{3},3\pi/2);\nonumber\\
	P_1^H(0)&=\Ibf-P_0^H(0)= P(\frac{4}{\sqrt{3}}, \frac{1}{\sqrt{3}}, \sqrt{3}, \pi/2);\  &P_1^H(1)=\Ibf-P_0^H(1)=P(\frac{4}{\sqrt{3}}, \sqrt{3},\frac{1}{\sqrt{3}},  \pi/2).
\end{align}
\end{enumerate}
Notice that 
\[  P(\mu, a, b, \theta)^2=\frac{1}{\mu^2}\begin{pmatrix}
	(1+a^2) & e^{-\iota \theta}(a+b)\\
	e^{\iota \theta}(a+b) & (1+b^2)
\end{pmatrix}, \] 
so that due to \eqref{eq:squareequal} 
$P(\mu, a, b, \theta)$ is a projector. Also notice that $\Ibf-P(\mu, a, b, \theta)= P(\mu, \mu-a, \mu-b, \theta\pm\pi)$ is also a projector since $(\Ibf-P(\mu, a, b, \theta))^2=\Ibf-2P(\mu, a, b, \theta)+P(\mu, a, b, \theta)=\Ibf-P(\mu, a, b, \theta)$.

Towards evaluation of the associated occupation measures $Q(u_B, u_H|\xi_B, \xi_H)$, first consider for $P(\mu, a, b, \theta)\in\Hscr_B$ and $P(\mu', a', b', \theta')\in\Hscr_H$,
\begin{equation}
	P(\mu, a, b, \theta)\otimes P(\mu', a', b', \theta')=\frac{1}{\mu\mu'}\begin{pmatrix}
aa'& a e^{-\iota\theta'}& a'e^{-\iota\theta}& e^{-\iota(\theta+\theta')}\\
ae^{\iota\theta}& ab'& e^{\iota(\theta'-\theta)}& b'e^{-\iota\theta}\\
a'e^{\iota\theta}& e^{\iota(\theta-\theta')}& a'b& b e^{-\iota\theta}\\
e^{\iota(\theta+\theta')}& b'e^{\iota\theta}& b e^{\iota \theta'} & bb'
	\end{pmatrix}
\end{equation}
so that, for finite $\mu,a,b$
\begin{equation}
\Tr\left(\rho P\left(\mu, a, b, \theta\right)\otimes P\left(\mu', a', b', \theta'\right)\right)=\frac{1}{\mu\mu'}\left(\frac{aa'+bb'}{2}+ \cos(\theta+\theta')\right).
\end{equation}


This allows us to tabulate all occupation measures corresponding to the specified quantum strategy $Q$. For instance,
\begin{multline}
	Q(0,0|0,0)=\Tr(\rho P(\infty, \infty, 0, 0)\otimes P(\frac{4}{\sqrt{3}}, \sqrt{3}, \frac{1}{\sqrt{3}}, \frac{3\pi}{2}))\\
	=\lim_{\mu\to\infty} \Tr (\rho P(\mu,\mu,0,0)\otimes P(\frac{4}{\sqrt{3}}, \sqrt{3}, \frac{1}{\sqrt{3}}, \frac{3\pi}{2}) ) = \lim_{\mu\to\infty} \frac{\sqrt{3}}{4\mu}\left(\frac{\mu \sqrt{3}}{2}-1\right)=\frac{3}{8}
\end{multline}
	\begin{table}[h]
		\centering
\begin{tabular}{ |c|c| } 
	\hline
{	$\mathbf{Q(u_B, u_H|\xi_B, \xi_H)}$ }& \textbf{$\mathbf{Q(u_B, u_H|\xi_B, \xi_H)}$} \\
	\hline
$Q(0, 0| 0,0)=3/8$ & $Q(0,0|1,0)=1/8$\\
\hline
$Q(0, 1| 0,0)=1/8$ & $Q(0,1|1,0)=3/8$\\
\hline
$Q(1, 0| 0,0)=1/8$ & $Q(1,0|1,0)=3/8$\\
\hline
$Q(1, 1| 0,0)=3/8$ & $Q(1,1|1,0)=1/8$\\
\hline
$Q(0, 0| 0,1)=1/8$ & $Q(0,0|1,1)=0$\\
\hline
$Q(0, 1| 0,1)=3/8$ & $Q(0,1|1,1)=1/2$\\
\hline
$Q(1, 0| 0,1)=3/8$ & $Q(1,0|1,1)=1/2$\\
\hline
$Q(1, 1| 0,1)=1/8$ & $Q(1,1|1,1)=0$\\
\hline
\end{tabular}
\caption{Occupation Measures of the specified quantum strategy $Q$}
\label{tab:quantoccmea}
	\end{table}

\begin{proposition}\label{prop:quantumadvantage}
	For the decision problem $D=(\lambda_B, \lambda_H, \chi(0), \chi(1))$, let $\lambda_B=\lambda_H=\nicefrac{4}{5}$, $\chi(0)=1$, and $\chi(1)=3$. Then the following strict inequality holds,
	\begin{equation}
		J_{\Qscr}^*(D)<J^*_\Lscr(D).
	\end{equation}
\end{proposition}
\begin{proof}
	We show that the strategy $Q\in \Qscr$ formulated above in Section \ref{sec:numdemo} beats all deterministic strategies strictly \ie $J(Q;D)< J^*_\Lscr(D)$.  Notice from Table \ref{tab:quantoccmea} that
	\begin{align*}
		J(Q;D)&= \sum_{\xi, u, \xi_W}\Pbb(\xi, \xi_W) \ell(u, \xi_W)Q(u|\xi)\\
		&=- \sum_{\xi, u, \xi_W}\Pbb(\xi, \xi_W) \chi(\xi_W) \delta(\xi_W, u_B\oplus u_H)Q(u|\xi)\\
		&=-\sum_{\xi}\Pbb(\xi, 0) (Q(0,0|\xi)+ Q(1,1|\xi)) + 3\Pbb(\xi, 1) (Q(1,0|\xi)+ Q(0,1|\xi))\\
		&= -\Pbb(\xi_W=0)-\sum_{\xi} (3\Pbb(\xi, 1)-\Pbb(\xi, 0))(Q(1, 0|\xi)+ Q(0, 1|\xi))\\
			&=-\frac{1}{2}-\frac{1}{8}\left(3\left(1-\left(\frac{4}{5}\right)\right)^2-\left(\frac{4}{5}\right)^2\right)-\frac{6}{8}\left(2.\frac{4}{5}.\left(1-\frac{4}{5}\right)\right)-\frac{1}{2}\left(3\left(\frac{4}{5}\right)^2-\left(1-\frac{4}{5}\right)^2\right)\\
			&= -\frac{1}{2}+\frac{13}{200}-\frac{48}{200}-\frac{188}{200}=-\frac{323}{200}< -\frac{320}{200}=-\frac{8}{5}=J_\Lscr^*(D)
	\end{align*}
The proposed now simply follows since $Q\in\Qscr$.
\end{proof}
The presented strategy can in fact be  implemented using a pair of entangled electrons. The state $\rho$ is a well known Bell basis state of such a pair \cite{neilsen2004qcqi} and local spin measurements on the electrons along suitably chosen axes in the three dimensional physical space implement the required measurement strategy. The measuement of outcome of the spin measurement, $s=\pm 1$, can be mapped to an action $u_i=(-1)^{(s+1)/2}$ to then specify an action strategy, to ultimately realize the quantum strategy that produces an advantage in Proposition \ref{prop:quantumadvantage}.

\subsection{Information, cost and the quantum advantage}
For $D=(\lambda_B, \lambda_H, \chi(0), \chi(1))$, recall from \eqref{eq:priordist} that the parmeter $\lambda_i$ captures the reliability or the quality of agent $i$'s information. On the other hand, the ratio $\chi(0)/\chi(1)$ determines the relative importance of guessing a particular value of $\xi_W\in\{0, 1\}$ correctly. In this section, we vary these parameters and investigate how the advantage offered by quanutm strategies behaves with this variation.  We discuss the insight that our investigation reveals about the nature of quantum strategies.

The classical optimum $J_\Lscr^*(D)$ is relatively simple to obtain using Proposition \ref{prop:classicalsol}. Computing the quantum infimum requires a search over a non-compact space $\Qscr$ for a general decision problem (see Section \ref{sec:openconvex}). However, owing to Tsirelson's rather detailed development \cite{tsirelson_1, tsirelson_2}, this search is substantially simple for a binary problem (two agents, binary actions and observations for each agent). In this case, $\Qscr$ is convex and compact, whence $J_\Qscr^*(D)$ is attained at an extreme point \cite{tsirelson_1, tsirelson_2, tsirelson_3}. In addition, it suffices to search over the compact space of projectors $\{P_{u_i}^i(\xi_i)\}_{i, u_i, \xi_i}$ while the entangled state $\rho$ is fixed at \eqref{eq:numstatematrix} to sift across all extreme points of $\Qscr$. This is how we generate our numerical results for plotting $J^*_\Qscr(D)$ across our variations. While we skip formal details of this discussion in this article, our subsequent work will carry the relevant theoretical proofs pertaining to the sufficiency of this search.

\subsubsection{Agent information and the quantum advantage}
We vary the reliabilities of agent's observations and investigate the correpsonding variation in the quantum advantage while $\chi(0)=1, \chi(1)=3$ are fixed. We present two plots in Figure \ref{fig:varyinfo}. In the first plot, we maintain $\lambda_B=\lambda_H=\lambda$ and vary $\lambda$ from $0.5$ to $1$. We observe a cutoff $\lambda_c>0.5$ which has to be exceeded by $\lambda$ for the quantum advantage to start appearing. However, the quantum advantage is maintained right upto the trivial upper cutoff $\lambda=1$ for this variation.

In the second plot, we set $\lambda_B=0.65$ and vary $\lambda_H$ from $0.65$ right upto $1$. We present here, both the no-signalling and the quantum advantage. First, we take note of the fact that the quantum advantage appears within an interval $(\underline\lambda_H, \overline\lambda_H)$ while the no-signalling advantage persists through a larger interval $[0.65,\widetilde\lambda_H)$ where $0.65<\underline\lambda_H<\overline\lambda_H<\widetilde\lambda_H$. Instances characterised by $\lambda_H\in (\overline\lambda_H, \widetilde\lambda_H)$ hence correspond to objectives that obey $J_\Lscr^*(D)=J_\Qscr^*(D)>J_{\Nscr\Sscr}^*(D)$. These have been  depicted in the second part of Figure \ref{fig:sliceplots}.

The existence of a non-trivial $\underline\lambda_H$ hints that a quantum advantage appears only once the `collective' information of the agents surpasses some threshold. On the other hand, to intuitively understand an existence of $\overline\lambda_H$ and $\widetilde\lambda_H$, we make a note that there is some randomness inherently within the no-signalling strategies outside $\Lscr$. It happens that once the agent $H$ knows `enough', he is better off estimating $\xi_W$ alone, as against randomly collaborating with $B$.
\begin{figure}[h]
	\centering
	\includegraphics[scale=0.23]{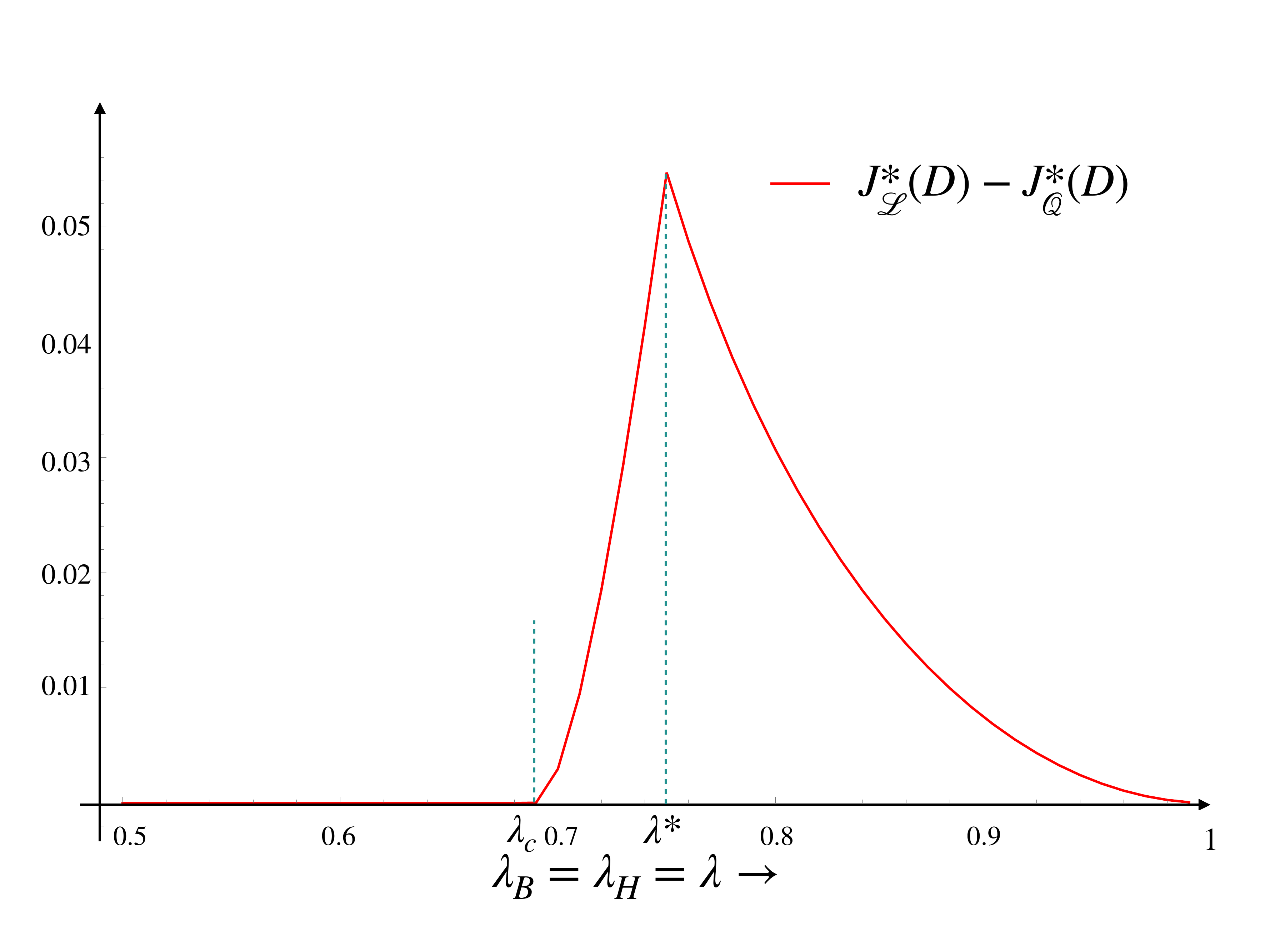}
	\includegraphics[scale=0.25]{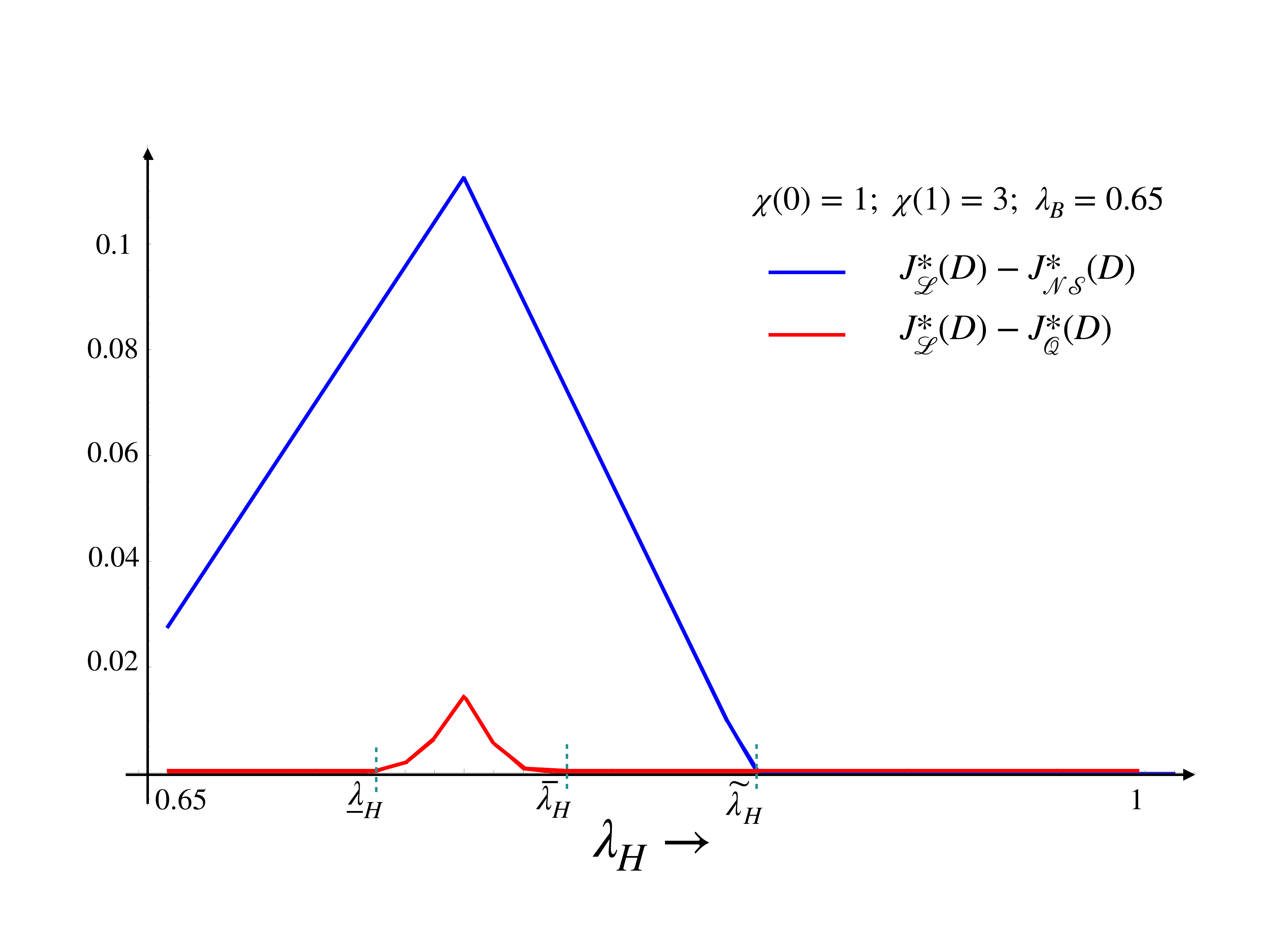}
	\caption{Variation of the Quantum Advantage with reliability of agent's observations.}
	\label{fig:varyinfo}
\end{figure}
A theoretical characterization of these cut-offs is of some interest but requires substantial development, which is the part of our ongoing work.

\subsubsection{Estimation dilemma and the inherent non-local randomness}
To motivate our investigation here, notice that if $\chi(1)=0$, then the algebraic optimum of the cost is simply obtained by $u_B\equiv u_H=0$ which is also the deterministic optimum. This is essentially because there is no reward for estimating $\xi_W=1$ correctly. Similarly in the limit $\chi(1)/\chi(0)\to \infty$, it is optimal for the decision makers to play $u_B\equiv 0,  u_H\equiv 1$ since it is infinitely more rewarding to guess $\xi_W=1$ correctly. We can then ask if there exists an intermediate interval for $\chi(1)/\chi(0)$ that contains all of the quantum advantage in our instances. The plot presented in Figure \ref{fig:chiadv} settles this query affirmatively.
\begin{figure}[h]
	\centering
	\includegraphics[scale=0.3]{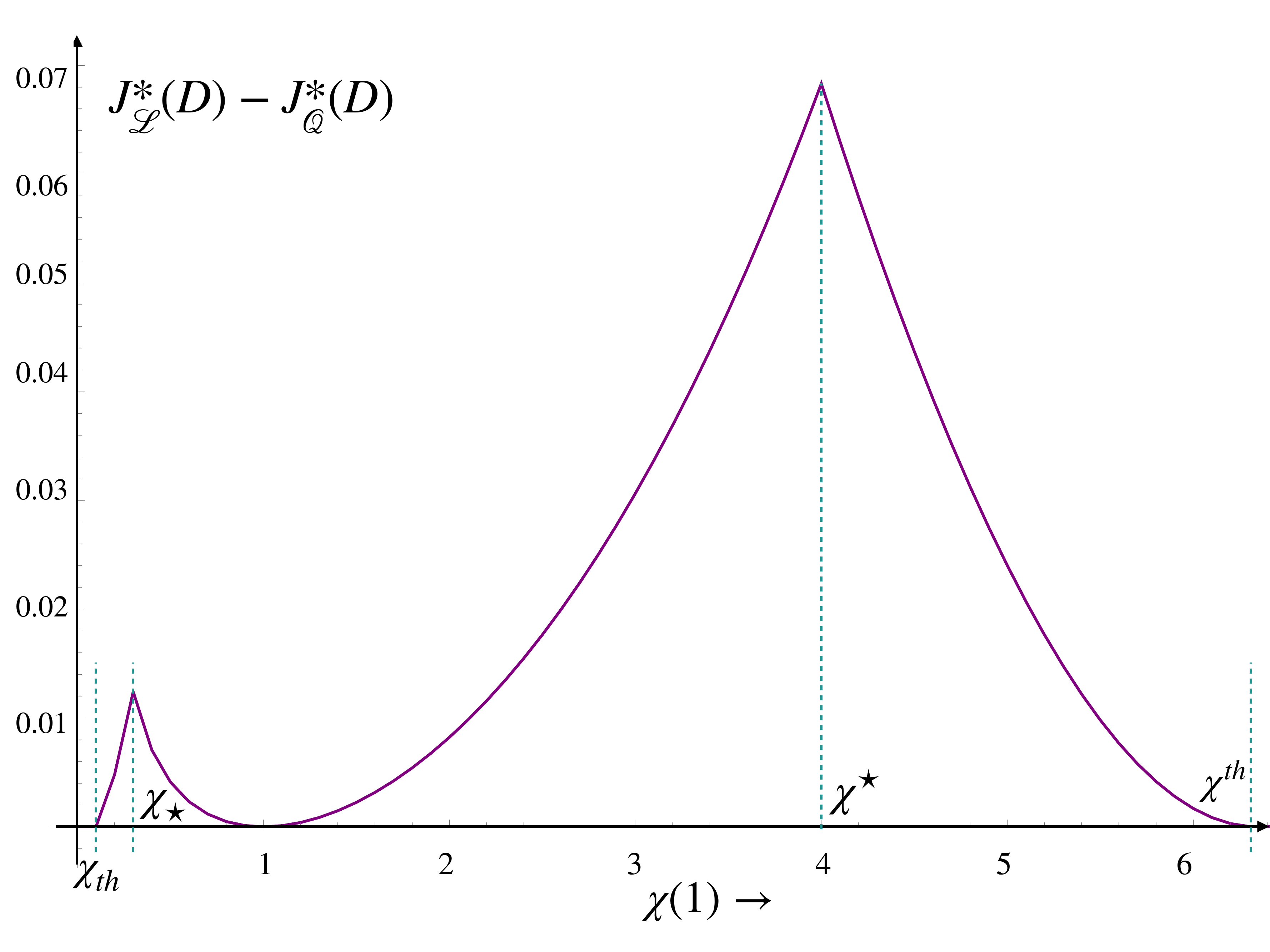}
	\caption{Variation of the Quantum Advantage with $\chi(1)$ $(\chi(0)=1, \lambda_B=\lambda_H=0.8)$}
	\label{fig:chiadv}
\end{figure}

We find that there are non-trivial lower and upper cut-offs, namely $\chi_{th}\leq 1$ and $\chi^{th}\geq 1$ respectively, on the ratio $\chi(1)/\chi(0)$ so that the quantum advantage is contained within their interval $(\chi_{th}, \chi^{th})$ as demonstrated in Figure \ref{fig:chiadv}.  Moreover, $\chi(0)=\chi(1)$ does not admit a quantum advantage either. The former observation requires substantial development for an analytical demonstration, and is the subject of our upcoming work \cite{deshpande2023binary2}. One can however intuitively explain this as an attribute of the inherent randomness in quantum strategies, \ie, once $\chi(1)/\chi(0)$ goes above $\chi^{th}$, the risk of incorrectly estimating $\xi_W=1$ brought by the randomness is no longer worth the improvement brought in by stronger correlations.  It is also worth noting that the advantage peaks at $\chi_\star\in (\chi_{th}, 1)$ and $\chi^{\star}\in (1, \chi^{th})$. It can be quickly seen why $\chi=1$ fails to admit a quantum advantage hence bifurcating our `region of interest' from the short calculation that follows. \\
Take any non-local vertex $Q^{\alpha\beta\delta}$ of $\Nscr\Sscr$ and notice (Recall \eqref{eq:nsvertex})
\begin{align*}
&J(Q^{\alpha\beta\delta};D)=-\sum_{\xi_B, \xi_H} \Pbb(\xi_B, \xi_H, 0) (\sim\xi_B.\xi_H\oplus \alpha \xi_B\oplus \beta \xi_H\oplus\delta)-\sum_{\xi_B, \xi_H} \Pbb(\xi_B, \xi_H, 1) (\sim\xi_B.\xi_H\oplus \alpha \xi_B\oplus \beta \xi_H\oplus\delta)\\
&=-(1/2)(\lambda^2 +(1-\lambda)^2)(\sim\delta+ \alpha\oplus\beta\oplus\delta)-\lambda(1-\lambda)(\sim\alpha\oplus\delta+\sim\beta\oplus\delta)\geq -\lambda=J_\Lscr^*(D)
\end{align*}
where the last inequality can be reasoned as follows. If $\delta=1$, then $\alpha=\beta=1$ minimizes the rest by ensuring $\alpha\oplus\beta\oplus\delta=\sim\alpha\oplus\delta=\sim\beta\oplus\delta=1$. This yields $J(Q^{110};D)=-1/2-\lambda(1-\lambda)\geq -\lambda$ for all $\lambda\in (1/2, 1]$. Otherwise if $\delta=0$, then $\alpha=1, \beta=0$ minimizes the sum and yeilds $J(Q^{100};D)=-1/2-(1/2)(-\lambda^2-(1-\lambda)^2)\geq -\lambda$ for all $\lambda\in (1/2, 1]$. Hence for $\chi(0)=\chi(1)=1$, $J_{\Nscr\Sscr}^*(D)=J_{\Qscr}^*(D)=J_{\Lscr}^*(D)$ holds.

\subsection{General properties of the quantum strategic space $\Qscr$}\label{generalstatic}
We now present some general properties of the Quantum strategic space $\Qscr$. In particular, it is relatively simple to establish the convexity of $\Qscr$ following \cite{werner2001multipartitebell}. It then follows that all strategies implementable using common randomness lie within $\Qscr$. We also show that quantum strategies respect the stasis of the information strucutre by establishing the inclusion $\Qscr\in\Nscr\Sscr$. This inclusion physically corresponds to the wave function collapse respecting the absence of super-luminal signalling as dictated by the theory of special relativity. Ultimately, the space of quantum strategies remains an open bounded `balloon' sandwiched within the local polytope $\Lscr$ and the no-signalling polytope $\Nscr\Sscr$. 
\subsubsection{Openness and convexity}\label{sec:openconvex}
The openness of $\Qscr$ was an open problem for a long time, and was recently shown in \cite{slofstra2019notclosed}. The convexity is rather well-known, and we establish it in the following proposition following \cite{werner2001multipartitebell}.
\begin{proposition}\label{prop:convexity}
	$\mathcal{Q}$ is convex. 
\end{proposition}
\begin{proof}
	Let 	$ {Q}_{\alpha}=(\{\Hscr^\alpha_i\}_i, \rho_\alpha, \{P_{u_i}^{i, \alpha}(\xi_i)\}_{i, u_i, \xi_i})\in \mathcal{Q}$ for each  $ \alpha\in \{1,\hdots,r\}$.
	Consider a strategy $Q$ which is a convex combination of $\{Q_\alpha\}$, \ie,
	\begin{equation}\label{eq:convexcomb}
		Q(u|\xi)=\sum_{\alpha}\theta_\alpha Q_{\alpha}(u|\xi)\ \forall\ u, \xi.
	\end{equation}
	We show that ${Q}\in \mathcal{Q}$. Let $Q=(\{\Hscr_i\}_i. \rho, \{P_{u_i}^i(\xi_i)\}_{i, u_i, \xi_i})$ where
	\begin{enumerate}[i)]
		\item 	$\mathcal{H}_{i}=\bigoplus_{\alpha=1}^r \mathcal{H}_{i}^{\alpha}$ whence 
		\begin{equation}
			\mathcal{H}:=\bigotimes_{k=1}^n\bigoplus_{\alpha=1}^r  \mathcal{H}^{\alpha}_{i}\cong\bigoplus_{\alpha_1=1}^r...\bigoplus_{\alpha_n=1}^r\bigotimes_i \mathcal{H}_{i}^{\alpha_i}.
		\end{equation} 
		\item $\rho=\bigoplus_{\alpha=1}^r\theta_\alpha\rho^\alpha$. Notice that $\rho$ satisfies 
		$\rho^\dagger=\rho$ by definition and
		$\Tr(\rho)=\sum_\alpha \theta_\alpha \Tr(\rho^\alpha)=\sum_\alpha \theta_\alpha=1$. Further since $\rho^\alpha\succeq 0$ and $\theta_\alpha>0$,
		$\rho\succeq 0$. Thus $\rho$ is a valid density operator on the following diagonal subspace of $\Hscr$:
		\begin{equation}
			\Hscr_d:=\bigoplus_{\alpha=1}^r\bigotimes_i\mathcal{H}^{\alpha}_{i}\subset \mathcal{H}.\label{diag}
		\end{equation}
		\item ${P}_{u_i}^{i}(\xi_i)=\bigoplus_{\alpha=1}^{r}{P}_{u_i}^{i, \alpha}(\xi_i)\in\Bscr(\Hscr_d)\subset\Bscr(\Hscr)$.
	\end{enumerate}
	We are now in a position to evaluate the joint conditional probability distribution over the action space given each agents' information for the strategy $Q$ specified by the Hilbert spaces $\Hscr^{(k)}\ \forall\ k$ and the projectors  $\boldsymbol{P}^{i}(u_i)(\xi_i) \forall\ i, \xi_{i}, u_i$. 
	\begin{equation}
		\begin{split}
			{Q}(u|\xi)&=\Tr\left(\bigoplus_{\alpha_1=1}^r P^{1, \alpha_1}_{u_1}(\xi_1)\otimes ...\otimes \bigoplus_{\alpha_n=1}^r P^{n, \alpha_n}_{u_n}(\xi_n) \rho  \right)
		\end{split}
	\end{equation}
	Substituting $\rho=\bigoplus_\alpha \theta_\alpha \rho^\alpha$,
	\begin{equation}\label{convcalc1}
		\begin{split}
			{Q}(u|\xi)&=\Tr\left(\bigoplus_{\alpha_1=1}^r...\bigoplus_{\alpha_n=1}^rP^{1,\alpha_1}_{u_{1}}(\xi_{1})\otimes ...\otimes P^{n,\alpha_n}_{u_{n}}(\xi_{n})\bigoplus_{\alpha=1}^r\theta_{\alpha}\rho^{\alpha} \right)
		\end{split}
	\end{equation}
	Notice that $\rho$ sits on the diagonal subspace \eqref{diag} and thus projections along off diagonal terms in the concatenation over $\alpha_1,...\alpha_n$ are null, i.e., 
	$$\Tr\left(P^{1,\alpha_1}_{u_{1}}(\xi_{1})\otimes ...\otimes P^{n,\alpha_n}_{u_{n}}(\xi_{n})\rho^{\alpha} \right)=0$$
	unless $\alpha_1=\alpha_2=...=\alpha_n=\alpha$. Thus, \eqref{convcalc1} simplifies to
	\begin{equation}\label{convcalc2}
		\begin{split}
			{Q}(u|\xi) =\sum_{\alpha}\theta_{\alpha}\Tr\left(P^{1,\alpha}_{u_{1}}(\xi_{1})\otimes ...\otimes P^{n,\alpha}_{u_{n}}(\xi_{n})\rho_{\alpha} \right)=\sum_{\alpha=1}^r\theta_{\alpha}Q_{\alpha}(u|\xi).
		\end{split}
	\end{equation}
	which in turn is consisitent with \eqref{eq:convexcomb} and we have thus established the convexity of $\mathcal{Q}$.
\end{proof}
\subsubsection{Inclusion of $\Lscr$}
\begin{proposition}
	$\Lscr\subset\Qscr$ and hence $J_{\Lscr}^*(D)\geq J_\Qscr^*(D)$.
\end{proposition}
\begin{proof}
	Since $\Qscr$ is convex from Proposition \ref{prop:convexity} and $\Lscr=\conv(\Gamma)$ from Theorem \ref{thm:detoptimum}, it suffices to show that $\Gamma\subset \Qscr$. So let $\gamma\in\Gamma$ be an arbitrary deterministic strategy. Take $Q=(\{\Hscr_i\}_i, \rho, \{P_{u_i}^i(\xi_i)\}_{i, u_i, \xi_i})$ with each $\Hscr_i$ as a unidimensional hilbert space $\Hscr$,  $\rho=(1)\in\Bscr(\Hscr)$ and $P_{u_i}^u(\xi_i)=\delta(u_i, \gamma_i(\xi_i))(1)\in\Bscr(\Hscr)$ for each $i, u_i$ and $\xi_i$. It is straightforward to verify that $Q\in\Qscr$. Also notice that 
	\begin{equation}
		Q(u|\xi)=\Tr(\rho\bigotimes_i P_{u_i}^i(\xi_i))=\prod_i \delta(u_i, \gamma_i(\xi_i))
	\end{equation}
	so that $Q\equiv \gamma$. Since $\gamma\in \Gamma$ was arbitrariliy chosen, it follows that $\Gamma\subset \Qscr$. The proposed is thus proved.
\end{proof}
\subsubsection{Inclusion within the no-signalling polytope}
\begin{proposition}
	All quantum strategies respect the no-signalling conditions, \ie, $\Qscr\subset\Nscr\Sscr$.
\end{proposition}
\begin{proof}
	To show this, consider an arbitrary $(\{\Hscr\}_i, \rho, \{P^i_{u_i}(\xi_i)\}_{i, u_i, \xi_i})=Q\in\Qscr$ and notice that it indeed obeys the no-signalling conditions \eqref{eq:nosignallingconst}:
	\begin{align*}
		\sum_{u_{-i}} Q(u_i, u_{-i}|\xi_i, \xi_{-i})&=\sum_{u_{-i}} \Tr(\rho\bigotimes_jP_{u_j}^j(\xi_j))\\
		&= \Tr(\rho P_{u_i}^i(\xi_i)\bigotimes_{j\neq i} \Ibf_j)
		= \sum_{u_{-i}} \Tr(\rho P_{u_i}^i(\xi_i)\bigotimes_{j\neq i}P_{u_j}^j(\xi'_j))= \sum_{u_{-i}} Q(u_i, u_{-i}|\xi_i, \xi'_{-i}) \ \forall u_i, \xi_{i}, \xi_{-i}, \xi_{-i}'.
	\end{align*}
	where the second equality follows from $\sum_{u_j} P_{u_j}^j(\xi_j)=\Ibf_j$ since $\{P_{u_j}^j(\xi_j)\}_{u_j}$ form a POVM. Hence $Q\in\Qscr \implies Q\in\Nscr\Sscr$ \ie $\Qscr\subset \Nscr\Sscr$.
\end{proof}

\subsubsection{Relative geometry of $\Lscr$, $\Qscr$ and $\Nscr\Sscr$, and its implications}
It is clear that $\Lscr\subset \Qscr\subset \Nscr\Sscr$ and existence of a strict quantum advantage amounts to the first inclusion $\Lscr\subset\Qscr$ being strict. In fact, the second inclusion $\Qscr\subset \Nscr\Sscr$ is also strict. This follows from each of the following two results in quantum information theory. First, it is known that $\Qscr$ is not a polytope \cite{goh2018geometry}. Second, the non-local vertices of $\Nscr\Sscr$ lie outside of $\Qscr$ \cite{ramanathan2016noquantum}. It also holds that  $\Lscr$ and $\Nscr\Sscr$ share some faces, which are essentially convex hulls over certain collections of the local vertices. A static team problem induces an optimisation of a linear objective over each of $\Lscr, \Qscr$ and $\Nscr\Sscr$ and its orientation determines if a quantum, or a no-singnalling advantage manifests. The following figure depicts the crux of our discussion through sliceplots \textit{representative} of the relative geometry of these strategic spaces.
\begin{figure}[h]
	\centering
	\includegraphics[scale=0.5]{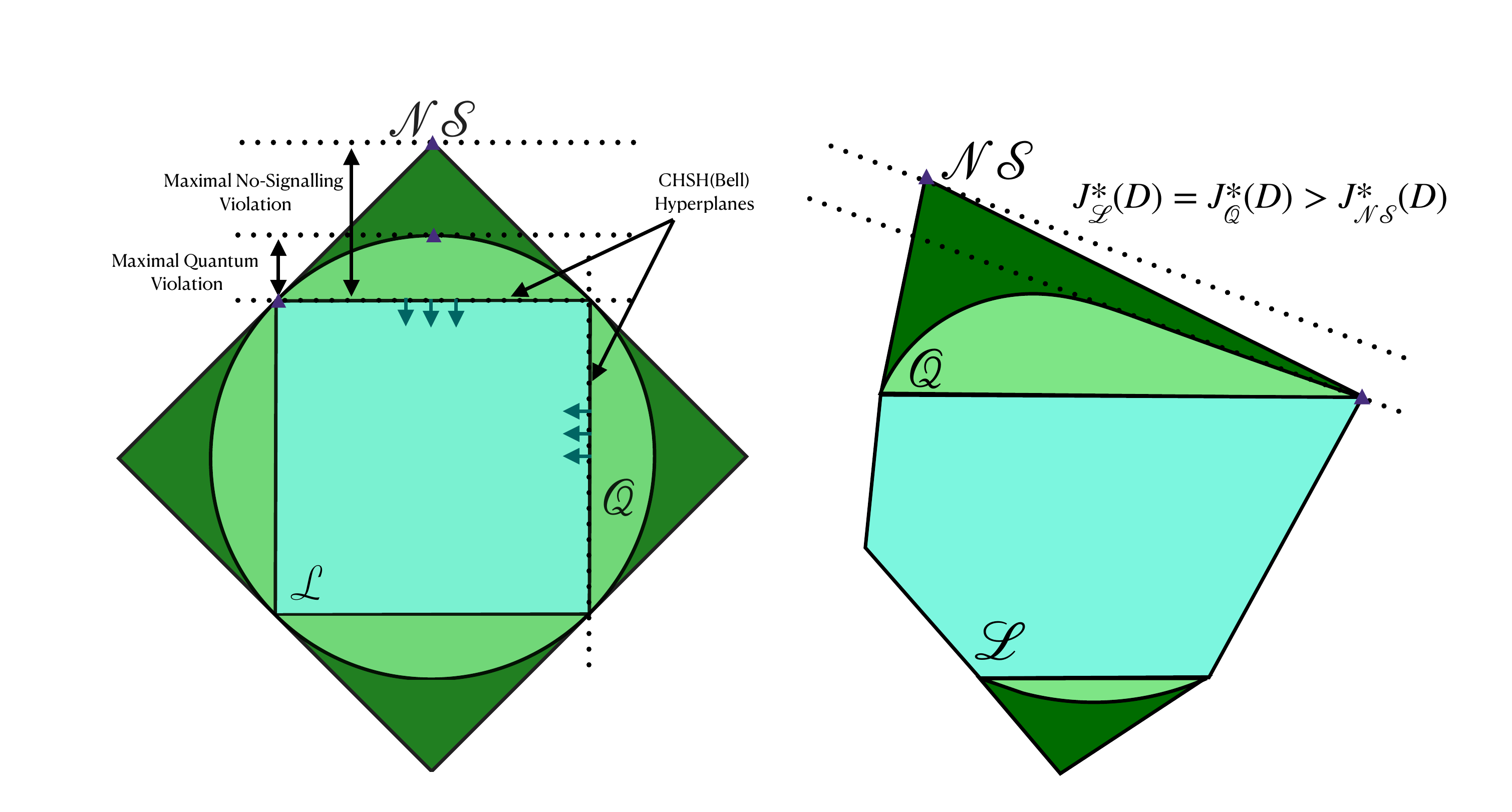}
	\caption{Representative sliceplots indicating relative geometry of the three strategic spaces}
	\label{fig:sliceplots}
\end{figure}

In Figure \ref{fig:sliceplots}, the former plot depicts the slice that shows the maximal quantum violation for two agents, with binary action and observation spaces. This slice features the popular problem known as the CHSH game which corresponds to the maximally violating objective \cite{clauser1969chsh}-- the CHSH objective with a rich history. This objective specifies a face of the local polytope $\Lscr$ and provides an orientation along which, the quantum optimum deviates maximally from the classical one. 
In the latter plot, we emphasize that the relative geometry of three spaces is rather exotic. In particular, there exist objectives oriented such that they admit a no-signalling advantage, and yet no quantum advantage.  Such problems have been exemplified by certain ranges of $\lambda$ in the second part of Figure \ref{fig:varyinfo}.

We conclude this subsection by noting that a general decision problem with a set of agents $\Nscr$, action and observation spaces $\Uscr_i, $ $\Xi_i$, a prior on observations $\Pbb: \prod_i\Xi_i\to [0, 1]$ and a cost $\ell: \prod_i \Uscr_i\times\Xi_i\to \Real$ admits an expected cost $J(Q;D)=\Jscr\t \Qbf$ is a linear objective with $\Qbf\in\Real^{\prod_i |\Xi_i||\Uscr_i|}$ as a vector with occupation measures $Q(u|\xi)$ as its components and $\Jscr$:
\begin{equation}
	\Jscr_{u, \xi}=\Pbb(\xi)\ell(u, \xi)
\end{equation}
determines the orientation of the cost objective in the strategic space.  Characterizing the prior distributions $\Pbb$ and the cost functions $\ell$ that induce an objective which admits a quantum advantage is thus of immediate interest. In the next section, we present a set of numerical results that demonstrate how the presence of a quantum advantage is swayed by parameters of decision-theoretic interest.

\section{Concluding Remarks}\label{sec:conclude}
We have demonstrated that in a decentralized estimation problem, the use of a quantum architecture allows a breach of limitations imposed upon the set of strategies generated through passive common randomness, while still respecting the information structure imposed by the control problem, a phenomenon we call the quantum advantage. In particular for an instance of decentralised estimation, it was observed that quantum mechanics allows for correlations among the agents' actions that more richly correlation actions of both agents, in contrast to strategies implemented using passive common randomness.

We have further demonstrated through a numerical investigation how elements of decision theoretic character such that the information  and the cost play a pivotal role in the appearance of a quantum advantage in the problem. The relative geometry of the pertaining strategic spaces was discussed to offer some geometric intuition to fluctuations in the presence of a quantum advantage. Our exposition has motivated the paradigm of entanglement assisted quantum strategies for decentralized control and decision-making in scenarios where cooperative decentalised agents have access to asymmetric information about the state of nature. 

%

\appendix

\appendix
\section{Quantum Mechanics}\label{app:qmprelim}
Following~\cite{preskillQIC}, we provide a preliminary introduction to the postulates of quantum mechanics, directly employed in our demonstration. While will work with the density matrix formalism of quantum mechanics, we also present here, the state-vector formalism for organic development.
\subsection{Hilbert Spaces and Operators}
A Hilbert space $\mathcal{H}$ is a vector space over the set of complex numbers $\mathbb{C}$. We use Dirac's bra-ket notation $\ket{\psi}\in \mathcal{H}$ to denote a vector in the Hilbert space $\mathcal{H}$. It carries an inner product $\braket{\psi|\phi}$ where $\bra{\psi}\in\mathcal{H}'$ denotes the co-vector of $\ket{\psi}$ and $\mathcal{H}'$ denotes the dual space of $\mathcal{H}$. The inner product maps an ordered pair of vectors to $\mathbb{C}$ and satisfies
\begin{enumerate}[i)]
	
	\item Positivity: $\braket{\psi|\psi}>0\ \forall \ket{\psi}\neq 0$
	
	\item Linearity: $\bra{\psi}(a\ket{\phi_1}+b\ket{\phi_2})=a\braket{\psi|\phi_1}+b\braket{\psi|\phi_2}\ \forall\ a,b\in\mathbb{C}$ 
	
	\item Skew Symmetry: $\braket{\phi|\psi}=\braket{\psi|\phi}^*$ where $^*$ denotes complex conjugation.
	
\end{enumerate}
A Hilbert space is complete in the norm $||\psi||=\braket{\psi|\psi}^{1/2}$ which we will work with hereon. An operator on $\Hscr$ is a map $\mathcal{H}\to\mathcal{H}$. We denote the space of operators on $\Hscr$ by $\Bscr(\Hscr)$. Consider a \textit{linear} operator $ \boldsymbol{M}\in\Bscr(\Hscr)$ and define its adjoint $\boldsymbol{M}^\dagger$ by
$$\braket{\psi|\boldsymbol{M}\phi}=\braket{\boldsymbol{M}^{\dagger}\psi|\phi}\ \ \forall \ket{\phi}, \ket{\psi}\in\mathcal{H}.$$
$\boldsymbol{M}$ is self-adjoint if $\boldsymbol{M}=\boldsymbol{M}^\dagger$. For each self-adjoint $\boldsymbol{M}$, the spectral theoren provides a spectral representation 
$$\boldsymbol{M}=\sum_{i}a_i\ket{i}\bra{i}$$
where $\{\ket{i}\}$ is an orthonormal basis in $\mathcal{H}$. Each $\ket{i}$ is then an eigenvector of the operator $\boldsymbol{M}$ with the corresponding real eigenvalue $a_i\in \Real$.

\subsection{Postulates of Quantum Mechanics}\label{qpostulates}
\noindent\textbf{Postulate 1. States of quantum systems:}
In the state vector formalism, a pure quantum mechanical state is described by a unit vector $\ket{\psi}\in\Hscr$ within some Hilbert space $\Hscr$. On the other hand, in the density matrix formalism of quantum mechanics, a state is a positive semi-definite operator $\rho$ on a complex Hilbert space with a unit trace and provides complete description of a physical system.  For different Hilbert spaces $\mathcal{H}_A$ and $\mathcal{H}_B$, we employ subscripts to denote vectors and density matrices in these spaces: $\ket{\psi_A}\in\Hscr_A$, $\ket{\psi_B}\in\Hscr_B$, $\rho_A\in\mathcal{H}_A$ and $\rho_B\in\mathcal{H}_B$. 

For each pure state $\ket{\psi}\in\Hscr$, we have a density matrix $\rho=\ket{\psi}\bra{\psi}\in\Bscr(\Hscr)$ that equivalently describes the same state. Note for consistency that $\rho\succeq 0$ by definition and $\Tr(\rho)=\braket{\psi|\psi}=1$. Conversely however, every density matrix $\rho\in\Bscr(\Hscr)$ does not admit a decomposition $\rho=\ket{\psi}\bra{\psi}$ for some $\ket{\psi}\in\Hscr$. Such density matrices are understood to represent mixed quantum states (as against pure quantum states). To understand mixed quantum states, first note that for each positive definite $\rho\in\Bscr(\Hscr)$ with a unit trace, the spectral theorem allows a decomposition
\begin{equation}
	\rho=\sum p_i \ket{\psi^i}\bra{\psi^i}=:\sum_i p_i\rho_i;\ p_i\geq 0\ \forall i;\ \sum_i p_i=1.
\end{equation}
where $\{\ket{\psi^i}\}$ is an orthonormal basis of $\Hscr$. We understand that such a $\rho$ represents a classically random state, distributed as $p_i=\Pbb(\rho_i)$ on the set $\{\rho_i\}_i$. This means that the systemic state is $\rho_i$ with probability $p_i$ for each $i$.


\noindent\textbf{Postulate 2. Observables of quantum systems:} An observable of a quantum system is a property that can be measured (in principle). An observable  $\boldsymbol{M}\in\Bscr(\Hscr)$ in quantum mechanics is a self-adjoint operator \ie $\boldsymbol{M}^\dagger=\boldsymbol{M}$. Suppose that for a basis $\{\ket{i}\}_i$ of $\Hscr$, we have  the spectral decomposition $\boldsymbol{M}=\sum_i a_i\ket{i}\bra{i}$. Let $\Sscr(\boldsymbol{M}):=\{a_j\}_j$ denote the set of eigenvalues of $\boldsymbol{M}$. Define $\Vscr(a):=\{\ket{i}:\boldsymbol{M}\ket{i}=a \ket{i}\}$ as the set of eigenvectors of an eigenvalue $a$. Now, to each $a_i\in\Sscr(\boldsymbol{M})$, associate a projection $P_i=\sum_{\ket{j}\in \Vscr(a_i)}\ket{j}\bra{j}$. Then, the spectral decomposition assumes the following form:
\begin{equation}
	\boldsymbol{M}=\sum_{a_i\in \Sscr(\boldsymbol{M})} a_i P_i.
\end{equation}
Note that since the normalised eigenvectors of $\boldsymbol{M}$ form an orthonormal basis of $\Hscr$, it follows that $\sum_i P_i=\Ibf$. The set $\{P_i\}_{i}$ is called a positive operator valued measure (\textit{POVM}) associated to the observable $\boldsymbol{M}$.

\noindent\textbf{Postulate 3. Measurements in quantum mechanics:}
Consider a measurement of an observable $\boldsymbol{M}\in\Bscr(\Hscr)$ on a systemic state $\rho\in\Bscr(\Hscr)$. Such a measurement results in an outcome $a_i\in\Sscr(\boldsymbol{M})$ with probability	 $\Pbb(a_i)=\Tr(\rho P_i)$ and the systemic state post measurement immediately collapses to 
$\rho_i:=P_i\rho P_i/\Tr(P_i\rho P_i)$. This abrupt transition into $\rho_i$ post measurement presents an exotic phenomenon of quantum mechanics known as the wave function collapse. The expected measurement outcome is then given by $\braket{a}=\sum_{a_i} a_i\Pbb(a_i)=\sum_{a_i} a_i\Tr(\rho P_i)=\Tr(\rho \boldsymbol{M})$.
In the case when $\rho=\ket{\psi}\bra{\psi}$ is pure, these relations translate to $\Pbb(a_i)=|\braket{P_i|\psi}|^2$, $\ket{\psi}_i=P_i\ket{\psi}/|P_i\ket{\psi}|$ where $\rho_i=:\ket{\psi}_{ii}\bra{\psi}$ and $\braket{a}=\braket{\psi|\boldsymbol{M}|\psi}$.

\noindent\textbf{Postulate 4. Dynamics of quantum systems:} With the parameter $t$ representing time parametrising the time evolution of a quantum state $\rho(t)$, the dynamics of a (closed) quantum system is specified by  the Schrödinger equation.  We have, the dynamics of the system specified as
\begin{equation}
	\rho(t')=U(t', t)\rho(t) U(t', t).
\end{equation}
By postulate 1, the operator $U(t', t)$ is unitary (since quantum states are unit vectors in Hilbert space, they evolve on a unit sphere) is given by
$U(t', t)=\exp{\left(-\int_t^{t'}\boldsymbol{H}(s)ds\right)}$where $\boldsymbol{H}(s)$ is a special self-adjoint operator known as the Hamiltonian of the system. Of course, for a pure $\rho(t)=\ket{\psi(t)}\bra{\psi(t)}$, the evolution can be specified in the state vector formalism as $\ket{\psi(t')}=U(t', t)\ket{\psi(t)}$.

\noindent\textbf{Postulate 5. Composition of quantum systems:} We begin within the state-vector formalism and briefly deal with mixed states later on.
Consider systems $A$ and $B$ respectively whose pure states are described by rays or unit vectors in Hilbert spaces $\mathcal{H}_A$ and $\mathcal{H}_B$. Then, pure states of the composite system $AB$ are described by rays in the Hilbert space $\mathcal{H}_{AB}$ given by the tensor product of the two spaces, $\mathcal{H}_{AB}=\mathcal{H}_A\otimes\mathcal{H}_B$. One immediate consequence of this postulate is that if system $A$ is in state $\ket{\phi_A}$ and system $B$ is in state $\ket{\psi_B}$ then the state of the composite system $AB$ is $\ket{\phi_A}\otimes\ket{\psi_B}$. However to understand the complete quantum mechanical picture of the composite system $AB$, we first understand the tensor product of the Hilbert spaces. Consider an $m$ dimensional orthonormal basis $\{\ket{i_A}\}$ for $\mathcal{H}_A$ and an $n$ dimensional orthonormal basis $\{\ket{\mu_B}\}$ for $\mathcal{H}_B$. Then the states $\ket{i_A, \mu_B}:=\ket{i_A}\otimes\ket{\mu_B}$ form an $m\times n$ dimensional orthonormal basis for $\mathcal{H}_{AB}$. The inner product in $\mathcal{H}_{AB}$ is defined by $\braket{i_A, \mu_B|j_A, \nu_B}=\delta_{ij}\delta_{\mu\nu}$
where $\delta_{ij}=1$ for $i=j$ and zero otherwise. Observables on the composite system are operators on $\mathcal{H}_{AB}$. Further, if an operator is rewritable as $\boldsymbol{T}_A\otimes\boldsymbol{V}_B$, it applies $\boldsymbol{T}_A$ on system $A$ and $\boldsymbol{V}_B$ on system $B$. Thus,
\begin{equation}
	\boldsymbol{T}_A\otimes\boldsymbol{V}_B\ket{i_A, \mu_B}=\boldsymbol{T}_A\ket{i_A}\otimes\boldsymbol{V}_B\ket{\mu_B}=\sum_{j, \nu}\ket{j_A, \nu_B}(T_A)_{ji}(V_B)_{\nu\mu}.
\end{equation}
Now in the density matrix formalism, a density matrix $\rho\in\Bscr(\Hscr_{AB})$ specifies a possibly mixed state of the composite system. The action of an operator $\boldsymbol{O}\in \Bscr(\Hscr_{AB})$ on this state of the system is then specified as $\boldsymbol{O}\rho\boldsymbol{O}$. Let us now look at how the behaviour of measurements on composite systems is postulated. 
\begin{enumerate}[a)]
	
	\item \textbf{Measurement on the composite system:} If an operator $\boldsymbol{M}_{AB}$ on $\mathcal{H}_{AB}$ is measured on the composite system in state
	$\ket{\psi_{AB}},$ postulate $3$ is directly applicable. Let $\{\ket{k_{AB}}\}_{k=1}^{mn}$ be the set of normalised eigenvectors of $\boldsymbol{M}_{AB}$, which span $\mathcal{H}_{AB}$, and we have a corresponding set of eigenvalues $\{a_k\}_{k=1}^{r}$ of $\boldsymbol{M}_{AB}$ such that by  the spectral decomposition (as in postulate 2), 
	$$\boldsymbol{M}_{AB}=\sum_{k}a_kP_k^{AB}.$$
	Then, upon measurement, the probability of measuring the eigenvalue $a_i$ is given by
	$$\Pbb(a_i)=||P_i^{AB}\ket{\psi_{AB}}||^2,$$ 
	and the composite system correspondingly collapses to the state
	$$\frac{P_i^{AB}\ket{\psi_{AB}}}{||P_i^{AB}\ket{\psi_{AB}}||}.$$
	Suppose that the measurement operator is rewritable as $\boldsymbol{M}_{AB}=\boldsymbol{M}_A\otimes\boldsymbol{M}_B$ where $\boldsymbol{M}_A$ and $\boldsymbol{M}_B$ are observables on $\mathcal{H}_A$ and $\mathcal{H}_B$ respectively. Suppose that we have their respective spectral decompositions,
	$$\boldsymbol{M}_A=\sum_i a_i P^A_i;\ \boldsymbol{M}_B=\sum_\mu b_\mu P^B_\mu.$$
	We can thus express,
	\begin{equation}
		\boldsymbol{M}_{AB}=\sum_{i, \mu} a_i b_\mu P^A_i\otimes P^B_\mu.\label{compdecomp}
	\end{equation}
	In physical situations involving such composite measurements, it may be possible to separately infer eigenvalues $a_i, b_\mu$ upon measurement. Thus we have the joint probability of measurement,
	\begin{equation}
		\Pbb(a_i, b_\mu)= ||P^A_i\otimes P^B_\mu\ket{\psi_{AB}}||^2\label{compprob}
	\end{equation}
	and the system thus collapses to the state
	\begin{equation}
		\ket{\psi_{AB}}^{i, \mu}=\frac{P^A_i\otimes P^B_\mu\ket{\psi_{AB}}}{||P^A_i\otimes P^B_\mu\ket{\psi_{AB}}||}.\label{compcol}
	\end{equation}
	Correspondingly for a mixed state $\rho_{AB}\in\Bscr(\Hscr_{AB})$, we have the probability of measurements and and the post measurement state given by
	\begin{equation}
		\Pbb(a_i, b_\mu)= \Tr(\rho_{AB} P_i^A\otimes P_\mu^B);\ \rho_{AB}^{i, \mu}=\frac{P_i^A\otimes P_\mu^B \rho_{AB} P_i^A\otimes P_\mu^B}{\Tr(P_i^A\otimes P_\mu^B \rho_{AB} P_i^A\otimes P_\mu^B)}
	\end{equation}
	
	\item \textbf{Measurement on a subsystem:} Suppose that the composite system is in state $\ket{\psi_{AB}}$ and an operator $\boldsymbol{M}_A$ on $\mathcal{H}_A$ is measured. Let $\{P_i^A\}_{i=1}^m$ be the \textit{POVM} corresponding to $\boldsymbol{M}_A$, and we have a corresponding set of eigenvalues $\{a_i\}_{i=1}$ so that by spectral theorem, we have the following decomposition for $\boldsymbol{M}_A$,
	$$\boldsymbol{M}_A=\sum_i a_i P_i^A.$$
	Measurement of $\boldsymbol{M}_A$ on subsystem $A$ amounts to a measurement of $\boldsymbol{M}_{AB}=\boldsymbol{M}_A\otimes\boldsymbol{I}_B$ on the composite system. Since the \textit{POVM} corresponding to the identity operator $\boldsymbol{I}_B$ is the singleton set $\{\boldsymbol{I}_B\}$ itself, we have the following decomposition of $\boldsymbol{M}_{AB}$ from (\ref{compdecomp}):
	$$\boldsymbol{M}_{AB}=\sum_{i} a_i P^A_i\otimes I_B.$$
	Such a measurement then results in an outcome $a_i$ with probability 
	$$\Pbb(a_i)=||P_i^A\otimes\boldsymbol{I}_B\ket{\psi_{AB}}||^2$$
	and the system correspondingly collapses to the state
	$$\frac{P_i^A\otimes\boldsymbol{I}_B\ket{\psi_{AB}}}{||P_i^A\otimes\boldsymbol{I}_B\ket{\psi_{AB}}||}.$$
	On the other hand, in the density matrix formalism, with a mixed state $\rho_{AB}\in\Bscr(\Hscr_{AB})$, the above measurement probabilities and the collapsed state post measurement is given by
	\begin{equation}
		\Pbb(a_i)=\Tr(\rho_{AB} P_i^A\otimes \boldsymbol{I}_B);\ \rho_{AB}^i=(\rho_{AB} P_i^A\otimes I_B)/\Tr(\rho_{AB} P_i^A\otimes I_B).
	\end{equation}
	
	\item \textbf{Sequential measurements on the two subsystems:} Consider a measurement of operator $\boldsymbol{M}_A$ on $\mathcal{H}_A$ immediately followed by a measurement of the operator $\boldsymbol{N}_B$ on $\mathcal{H}_B$ (This amounts to a measurement of  $\boldsymbol{M}_A\otimes \boldsymbol{I}_B$ immediately followed by a measurement of  $\boldsymbol{I}_A\otimes \boldsymbol{N}_B$ on the composite system). We show that this is equivalent to the measurement of $$\boldsymbol{M}_A\otimes \boldsymbol{N}_B=\boldsymbol{M}_A\otimes \boldsymbol{I}_B\times \boldsymbol{I}_A\otimes \boldsymbol{N}_B$$ on the composite system. Suppose that the operators can respectively decomposed spectrally as
	\begin{equation}
		\boldsymbol{M}_A=\sum_i a_i P_i^A,\ 
		\boldsymbol{N}_B=\sum_\mu b_\mu P_\mu^B.\label{Opdecomp}
	\end{equation}
	The first measurement then results in an outcome $a_i$ with probability
	$$\Pbb(a_i)=||P_i^A\otimes\boldsymbol{I}_B\ket{\psi_{AB}}||^2$$
	and the system correspondingly collapses to the state
	$$\ket{\phi_{AB}}=\frac{P_i^A\otimes\boldsymbol{I}_B\ket{\psi_{AB}}}{||P_i^A\otimes\boldsymbol{I}_B\ket{\psi_{AB}}||}.$$
	A similar account for the mixed state $\rho_{AB}$ is provided through
	\begin{equation}\label{eq:partialmeadensity}
		\Pbb(a_i)=\Tr(\rho_{AB} P_i^A\otimes \boldsymbol{I}_B);\ \rho_{AB}^i=(\rho_{AB} P_i^A\otimes I_B)/\Tr(\rho_{AB} P_i^A\otimes I_B).
	\end{equation}
	The immediate, second measurement then takes place on this collapsed state, following which the probability of outcome $b_\mu$ is
	$$\Pbb^i(b_\mu)=||\boldsymbol{I}_A\otimes P_\mu^B\ket{\phi_{AB}}||^2=\frac{||\boldsymbol{I}_A\otimes P_\mu^B\times P_i^A\otimes\boldsymbol{I}_B\ket{\psi_{AB}}||^2}{||P_i^A\otimes\boldsymbol{I}_B\ket{\psi_{AB}}||^2}=\frac{||P_i^A\otimes P_\mu^B\ket{\psi_{AB}}||^2}{||P_i^A\otimes\boldsymbol{I}_B\ket{\psi_{AB}}||^2}.$$
	Thus the joint probability of measuring the pair of eigenvalues $a_i, b_\mu$ is
	\begin{equation}
		\Pbb(a_i, b_\mu)=\Pbb(a_i)\Pbb^i(b_\mu)=||P_i^A\otimes P_\mu^B\ket{\psi_{AB}}||^2\label{seqmeaprob}
	\end{equation}
	and the final state of the system after the two sequential measurements is
	\begin{equation}
		\frac{\boldsymbol{I}_A\otimes P_\mu^B\ket{\phi_{AB}}}{||\boldsymbol{I}_A\otimes P_\mu^B\ket{\phi_{AB}}||}=\frac{P_i^A\otimes P_\mu^B\ket{\psi_{AB}}}{||P_i^A\otimes P_\mu^B\ket{\psi_{AB}}||}.\label{seqmeacol}
	\end{equation}
	Notice that the measurement probability (\ref{seqmeaprob}) and the collapsed state of the composite system (\ref{seqmeacol}) after sequential measurements does correspond with the probability of composite measurement of the operator $\boldsymbol{M}_A\otimes\boldsymbol{N}_B$ (\ref{compprob}) and the corresponding collapsed state (\ref{compcol}). This establishes our claim that the sequential measurement is equivalent to the composite measurement of $\boldsymbol{M}_A\otimes\boldsymbol{N}_B$.
	A similar line of arguments carries over to the density matrix formalism to establish our claim. In particular, we evaluate using \eqref{eq:partialmeadensity}
	\begin{equation}
		\Pbb^i(b_\mu)=\Tr(\rho^i_{AB}\boldsymbol{I}_A\otimes P_\mu^B)=\Tr(\rho_{AB} P_i^A\otimes \boldsymbol{I}_B \cdot  \boldsymbol{I}_A\otimes P_\mu^B)/\Tr(\rho_{AB} P_i^A\otimes \boldsymbol{I}_B)
	\end{equation}
	so that $\Pbb(a_i, b_\mu)=\Pbb^i(b_\mu)\Pbb(a_i)=\Tr(\rho_{AB} P^A_i\otimes P^B_\mu)$, thus demonstrating our claim for mixed states in general, within the density matrix formalism. Consequently, generalisation of systemic composition to $n>2$ systems becomes rather immediate.
\end{enumerate}
These five postulate provide a complete formulation of quantum mechanics. We now look at an exotic phenomenon in quantum physics that arises from postulates 3 and 5. 

\subsection{Entanglement in Composite Quantum Systems.}
Let $\{\Hscr_i\}_i$ be the Hilbert spaces corresponding to quantum systems $\{i\}$. A pure quantum state $\ket{\psi}\in\bigotimes_i\Hscr_i=:\Hscr$ is entangled if there exist no $\ket{\psi_i}\in\Hscr_i$ such that $\ket\psi=\bigotimes_i\ket{\psi_i}$. Otherwise, $\ket{\psi}$ is entangled. A mixed state $\rho\in\Bscr(\Hscr)$ is not entangled if there exist $\rho_i\in\Bscr(\Hscr_i)$ such that $\rho=\bigotimes_i\rho_i$. A general claim for mixed states otherwise cannot be made with any simplicity.

To exemplify entanglement, assume orthonormal bases $\{\ket{+_A}, \ket{-_A}\}$ and $\{\ket{0_B}, \ket{1_B}, \ket{2_B}\}$  for spaces $\mathcal{H}_A$ and $\mathcal{H}_B$ so that the combined system $AB$ is described in $\mathcal{H}_{AB}$ with orthonormal basis $\{\ket{+_A, 0_B}, \ket{+_A, 1_B}$ $, \ket{+_A, 2_B}, \ket{-_A, 0_B},$ $\ket{-_A, 1_B}, \ket{-_A, 2_B}\}$. Consider the following state of the composite system $AB$
\begin{equation}
	\ket{\psi_{AB}}=\frac{1}{2}(\ket{+_A, 0_B}+\ket{-_A, 1_B}+\sqrt{2}\ket{-_A, 2_B})\label{entstateab}; \ \rho=\ket{\psi_{AB}}\bra{\psi_{AB}}.
\end{equation}
It can be shown with simple algebraic arguments that there exist no states $\ket{\phi_A}\in\mathcal{H}_A$ and $\ket{\chi_B}\in\mathcal{H}_B$ such that $\ket{\psi_{AB}}=\ket{\phi_A}\otimes \ket{\chi_B}$. We call $\ket{\psi_{AB}}$ an \textit{entangled state} of the composite system $AB$.  Quantum entanglement is a natural consequence of postulate 5, there exist quantum states $\sum_{i, \mu} u_{i\mu}\ket{i_A, \mu_B}\in \mathcal{H}_{AB}$ such that for all $\{w_i\}$ and $\{v_\mu\}$, we have 
$$\sum_{i, \mu} u_{i\mu}\ket{i_A, \mu_B}\neq\sum_i w_i\ket{i_A}\otimes\sum_\mu v_\mu\ket{\mu_B}.$$

\noindent\textit{Remark:} Suppose that an unentangled quantum state, $\rho_{AB}=\ket{\psi_{AB}}\bra{\psi_{AB}}$ separates as $\rho_A\otimes\rho_B$ with $\rho_i=\ket{\psi_i}\bra{\psi_i}$. 
Then any quantum strategy that rests upon such a state is behaviourally randomized:
\begin{equation}
	Q(u|\xi)=\Tr\left(\rho_{AB} P_{u_A}^A(\xi_A)\otimes P_{u_B}^B(\xi_B)\right)=\Tr\left(\rho_A  P_{u_A}^A(\xi_A)\right)\Tr\left(\rho_B P_{u_B}^B(\xi_B)\right)=Q(u_A, u_B|\xi_A, \xi_B).
\end{equation}

\bibliographystyle{unsrt}  

\bibliography{ref}

\end{document}